\documentclass[11pt]{article} %
\usepackage{pgfplots}
\pgfplotsset{width=7cm,compat=1.9}
\usepackage{booktabs}
\usepackage{float}
\usepackage[margin=1in]{geometry}
\usepackage[scr=rsfso]{mathalfa}
\usepackage{mdframed} %
\usepackage{microtype} %
\usepackage{mathpazo}
\usepackage{amsthm}
\usepackage{xspace}
\usepackage{comment}
\usepackage{lipsum}
\usepackage{amsfonts}
\usepackage{graphicx}
\usepackage{epstopdf}
\usepackage{algorithm}
\usepackage[noend]{algorithmic}
\usepackage{tikz}
\usepackage{caption}
\usepackage{subcaption}
\usepackage{amsmath}
\usepackage{hyperref}
\usepackage{amssymb}
\usepackage{dirtytalk}
\usepackage[colorinlistoftodos,textsize=tiny,textwidth=2cm,color=green!50!gray]{todonotes}
\usetikzlibrary{decorations.pathmorphing}
\ifpdf
\DeclareGraphicsExtensions{.eps,.pdf,.png,.jpg}
\else
\DeclareGraphicsExtensions{.eps}
\fi

\let\epsilon=\varepsilon %

\newcommand{\lat}{\ell}

\newcommand{\maxlat}{\ell_{\max}}
\newcommand{\maxlatrounded}{\bar{\ell}_{\max}}

\newcommand{\ub}{ub}

\newcommand{\atsp}{\alpha_{\textnormal{ATSP}}}
\newcommand{\atspp}{\alpha_{\textnormal{ATSPP}}}
\newcommand{\rootset}{V_{\textnormal{copy}}}
\newcommand{\roots}{R}
\newcommand{\eps}{\epsilon}
\newcommand{\target}{s'}

\newcommand{\deliveryperson}{\textsc{Deliveryperson}\xspace}
\newcommand{\repairperson}{\textsc{Traveling Repairperson}\xspace}
\newcommand{\tsp}{traveling salesperson problem\xspace}
\newcommand{\tspp}{\textsc{Asymmetric Path TSP}\xspace}
\newcommand{\probatsp}{\textsc{Asymmetric TSP}\xspace}

\def\opt{\mathop{\rm opt}\nolimits}
\def\OPT{\mathop{\rm OPT}\nolimits}
\def\group{\mathcal{G}}
\def\np{\mathop{\rm NP}\nolimits}

\newtheorem{theorem}{Theorem} 
\newtheorem{lemma}[theorem]{Lemma} %
\newtheorem{corollary}[theorem]{Corollary} %
\newtheorem{definition}[theorem]{Definition} %
\newtheorem{claim}[theorem]{Claim} %

\usepackage[inline]{enumitem}
\setlist[enumerate,1]{label=(\roman*)}
\setlist[enumerate,2]{label=(\alph*)}

\newlist{stepsroman}{enumerate}{1}
\setlist[stepsroman]{label=(\roman*)}

\newlist{stepsarabic}{enumerate}{1}
\setlist[stepsarabic]{label=\arabic*., ref=\arabic*}

\newlist{stepsalph}{enumerate}{1}
\setlist[stepsalph]{label=(\alph*), ref=(\alph*)}

\usepackage[capitalize, nameinlink]{cleveref}
\crefname{theorem}{Theorem}{Theorems}
\Crefname{lemma}{Lemma}{Lemmas}
\Crefname{claim}{Claim}{Claims}
\Crefname{fact}{Fact}{Facts}
\Crefname{remark}{Remark}{Remarks}
\Crefname{observation}{Observation}{Observations}
\Crefname{figure}{Figure}{Figures}
\Crefname{line}{Line}{Lines}
\Crefname{algocf}{Algorithm}{Algorithms}
\crefalias{AlgoLine}{line}
\Crefname{stepsromani}{Step}{Steps}
\Crefname{stepsarabici}{Step}{Steps}
\Crefname{stepsalph}{Step}{Steps}

\AddToHook{env/lemma/begin}{\crefalias{theorem}{lemma}}
\AddToHook{env/claim/begin}{\crefalias{theorem}{claim}}
\AddToHook{env/observation/begin}{\crefalias{theorem}{observation}}
\AddToHook{env/definition/begin}{\crefalias{theorem}{definition}}

\def\DEBUG{true} 
\ifdefined\DEBUG

\newcommand{\jbr}[1]{\todo[color=blue!100!black!50]{J: #1}}
\newcommand{\ramr}[1]{\todo[color=red!100!white!50]{R: #1}}
\else

\newcommand{\jbr}[1]{}

\newcommand{\ramr}[1]{}
\fi

\usepackage[
backend=biber,
style=alphabetic,
citestyle=alphabetic,
maxalphanames=6,
maxcitenames=99,
mincitenames=98,
maxbibnames=99,
giveninits=true,
url=false,
doi=true,
isbn=true,
backref=true
]{biblatex}

\begin{document}
	
	\title{A Constant-Factor Approximation for Directed Latency}
	
	\author{%
		Jannis Blauth%
		\footnote{%
			Institute for Operations Research, ETH Zurich, Zurich, Switzerland.
			Email:
			\href{mailto:jblauth@ethz.ch}{jblauth@ethz.ch}.
		}%
		\and%
		Ramin Mousavi%
		\footnote{%
			IDSIA, USI-SUPSI, Lugano, Switzerland.
			Email:
			\href{mailto:ramin.mousavi@supsi.ch}{ramin.mousavi@supsi.ch}. Partially supported by the SNF Grant 200021{\_}200731.
		}%
	}%
	
	\date{}
	
	\maketitle
	
	\begin{abstract}
		
		In the \textsc{Directed Latency} problem, we are given an asymmetric metric space $(V \cup \{s\},c)$ on a set $V$ of clients and a depot $s$.
		We are looking for a path $P$ starting in $s$ that visits all clients and minimizes the sum of the clients' waiting times (also known as latency) before being visited on the path. This models problems in logistics where client satisfaction is essential, as opposed to objectives like in \textsc{TSP}, where the goal is to make the salesperson as happy as possible. 
		
		In contrast to the symmetric version of this problem (also known as the \deliveryperson problem and the \repairperson problem in the literature), there are significant gaps in our understanding of \textsc{Directed Latency}. The best approximation factor has remained at $O(\log |V|)$, as shown by [Friggstad, Salavatipour, and Svitkina, '13], for more than a decade. Only recently, [Friggstad and Swamy, '22] presented a constant-factor approximation, but in quasi-polynomial time. Both results follow similar ideas: they consider buckets with geometrically increasing distances, build a path on each bucket, and then stitch together all these paths to get a feasible solution. Building a path on each bucket can be done cheaply thanks to developments in \textsc{Asymmetric Path TSP}. However, stitching these paths together incurs a logarithmic factor increase in the latency. [Friggstad and Swamy, '22] showed that by guessing a vertex from each bucket and augmenting a standard LP relaxation with these guesses, one can reduce the stitching cost. Unfortunately, the number of buckets is logarithmic in the number of vertices, so the running time of their algorithm is quasi-polynomial.
		
		In this paper, we present the first constant-factor approximation for \textsc{Directed Latency} in polynomial time by introducing a completely new way of bucketing, which helps us strengthen a standard LP relaxation with less aggressive guessing. 
		Although the resulting LP is no longer a relaxation of \textsc{Directed Latency}, it still admits a good solution.
		Then, we present a rounding algorithm for fractional solutions of our LP, which at a high level follows the rounding algorithm by [Friggstad and Swamy, '22] but with many new ingredients, crucially exploiting the way we restricted the feasibility region of the LP formulation.
	\end{abstract}
	\thispagestyle{empty}
	
	\newpage
	\setcounter{page}{1}
	
	\section{Introduction}\label{sec: intro} 
	
	The \textsc{Latency} problem is a variant of the \textsc{Path TSP} problem, where we are given a metric space on a vertex set $V \cup \{s\}$ with a specified starting point $s$, and the goal is to find a spanning path starting in $s$ that minimizes the sum of the distances of each vertex to $s$ on the path. The distance of vertex $v$ to $s$ on the path is called {\em latency} or {\em waiting time} of $v$. In contrast, \textsc{Path TSP} can be seen as minimizing the maximum vertex waiting time.
	It is known that the \textsc{Latency} problem is NP-hard.
	\footnote{It is possible to reduce the \textsc{TSP} where all distances are either $1$ or $2$ to the \textsc{Latency} problem.
		Indeed, given such a \textsc{TSP} instance $(V,c)$, one can declare an arbitrary vertex as starting point $s$, introduce $N$ colocated auxiliary vertices $v^*_1,\dots,v^*_N$, and set $c(v,v^*_i) = c(v,s)+M$ for each $v \in V$ and each $i \in \{1,\dots,N\}$.
		By choosing $M$ large enough (depending on $n$ and $N$), any optimum latency tour visits $v^*_1,\dots,v^*_N$ last.
		By choosing $N$ large enough (depending on $n$), any optimum latency tour yields an optimum tour on $V$.}
	\cite{blum1994minimum} gave the first constant-factor approximation for the \textsc{Latency} problem. 
	
	In literature, this problem is also referred to as the \deliveryperson problem \cite{del_man} and the \repairperson problem \cite{repairman}.
	These names are inspired by applied settings in which this problem naturally arises. In these settings, the actual route length plays a minor role, and the main goal is to minimize the average  time a customer needs to wait before being visited. 
	In another applied setting, we consider a vehicle whose travel cost scales with the load it carries. This problem can be viewed as an instance of the \textsc{Latency} problem. Indeed, in any solution to the \textsc{Latency} problem, the cost of the $i$'th edge of the tour is incurred with a factor of $|V|+1-i$.
	In addition to applications in logistics, the objective in the \textsc{Latency} problem is closely related to {\em searching a graph} (e.g., the web graph) for a hidden treasure \cite{koutsoupias1996searching,ausiello2000salesmen}. 
	We remark that sometimes, the \textsc{Latency} problem is stated to search for a tour instead of a path. Note that this results in an equivalent problem since we do not care about the time we need to return to the depot $s$.
	
	In this paper, we study a generalization of the \textsc{Latency} problem to asymmetric metrics. Formally, in the \textsc{Directed Latency} problem, we are given an asymmetric metric space $(V \cup \{s\}, c)$.
	The task is to compute a TSP-path $P$ starting in $s$ and visiting all vertices in $V$, so that $\lat(P) := \sum_{v \in V} c_P(v)$ is minimized, where for $v \in V$, $c_P(v)$ denotes the length of the unique $s$-$v$-subpath in $P$.
	We also call $c_P(v)$ the \emph{latency} of $v$ in $P$, and $\ell(P)$ the \emph{total latency} of $P$. 
	
	Network design and vehicle routing problems in directed graphs are often much harder to approximate than their undirected versions. The \textsc{TSP} and the \textsc{Steiner Tree} problems are well-known such examples. \textsc{Directed Latency} is no exception: there has been much progress in the \textsc{Latency} problem including a constant factor approximation (see \cref{subsec: related work}), but only a logarithmic approximation for \textsc{Directed Latency} is known \cite{logn_dirlat}. The question whether this problem admits a constant-factor approximation algorithm remained open since it was explicitly studied for the first time more than a decade ago \cite{nagarajan2008directed}. Only recently, \textcites{quasi-poly} made partial progress towards this question by presenting a constant-factor approximation in quasi-polynomial time.
	
	In this paper, we fully settle this question by providing the first constant-factor approximation for \textsc{Directed Latency}.
	
	\begin{theorem}\label{thm:main}
		There is a constant-factor approximation algorithm for \textsc{Directed Latency}.
	\end{theorem}
	
	Imagine we want to measure the time a client $v$ is waiting on path $P$ in excess of its least possible waiting time, i.e, $c_P(v)-c(s,v)$. This is the notion of the so-called regret metric, and  measures the dissatisfaction of clients. In the regret metric we want to visit the clients closer to $s$ earlier on our path which is neglected in the normal latency objective. This can be modeled as finding a spanning path $P$ that minimizes the total regret, i.e., $\sum_{v\in V}(c_P(v)-c(s,v))$. \textcites{quasi-poly} formulated this problem as \textsc{Minimum Total Regret}, and obtained a constant factor approximation in quasi-polynomial time 
	(with a better constant than for \textsc{Directed Latency} if the regret metric is obtained from a symmetric metric). Note that the regret metric, i.e., $c_{\textnormal{reg}}(u,v):=c'(s,u)+c'(u,v)-c'(s,v)$ for $u,v \in V \cup \{s\}$, forms an asymmetric metric where $c'$ is a symmetric or asymmetric metric on $V\cup\{s\}$. Therefore, this problem can be cast as a special case of \textsc{Directed Latency}. As a consequence of Theorem \ref{thm:main}, we also get a constant factor approximation for the \textsc{Minimum Total Regret} problem.
	
	\begin{corollary}\label{cor:by product}
		There is a constant-factor approximation algorithm for \textsc{Minimum Total Regret}.
	\end{corollary}
	
	Although our result is not relative to the optimal value of some linear programming (LP) relaxation, our algorithm heavily relies on a natural LP relaxation for \textsc{Directed Latency}, which was introduced for the symmetric version of this problem by \textcite{lp_undirected} and first used for the directed setting in \cite{quasi-poly}.
	On a very high level, we follow the idea of \textcites{quasi-poly}, who provide a constant-factor approximation for \textsc{Directed Latency} in time $n^{O(\log n)}$. 
	Here and henceforth, $n := |V|$ denotes the number of clients in the given instance of \textsc{Directed Latency}.
	Before solving the LP relaxation, we strengthen it by imposing constraints based on guesses about the structure of a fixed optimal integral solution.
	Although, in contrast to \cite{quasi-poly}, the optimal solution is no longer feasible for the strengthened LP, we can still show that there is a feasible solution whose objective value is within a constant factor of the optimal total latency.
	Finally, we show how to round any fractional solution to the LP, leveraging the new class of constraints introduced in this paper.
	
	\subsection{Further related work}\label{subsec: related work}
	In the asymmetric setting, if we instead aim for minimizing the maximum latency among all customers, this results in the asymmetric path version of the famous \tsp (\textsc{Asymmetric Path TSP}).
	\textcite{tspp_gap_kohn} showed an LP-relative constant-factor approximation for \tspp. Specifically, they achieved an upper bound of $4\cdot\alpha_{\textnormal{ATSP}}-3$ on the standard LP relaxation for \tspp, where $\alpha_{\textnormal{ATSP}}$ is the integrality gap of the Held-Karp relaxation for \probatsp. \textcite{first_constant_ATSP} were the first to prove that $\alpha_{\textnormal{ATSP}}=O(1)$. 
	We remark that, although not being LP-relative, the first constant-factor approximation for \textsc{Asymmetric Path TSP} was a direct consequence of \cite{first_constant_ATSP}, using a blackbox-reduction provided in \cite{ATSPP}.
	Recently, \textcite{vygenATSP} showed $\atsp < 15$. Moreover, \textcite{quasi-poly} showed that the integrality gap of a weakening of the standard LP relaxation, where non-$(s,t)$-cuts only need to be covered to the extent of $2\rho$ with $\rho \in (\frac12,1)$ (instead of at least $2$), still remains a constant (depending on $\rho$). This will be used as a subroutine in our algorithm.
	
	As mentioned earlier, computing a minimum latency path in undirected graphs has seen much progress both in designing better approximation algorithms and considering different restricted metrics. The first constant-factor approximation algorithm for the \textsc{Latency} problem was given by \cite{blum1994minimum} who provide an $8\alpha$-approximation, where $\alpha$ is the best approximation factor for \textsc{$k$-MST}. Subsequently, \textcite{goemans1998improved} improved this to roughly $3.59\cdot\alpha$, and finally \textcite{chaudhuri2003paths} removed the dependence on $\alpha$ which is currently the best approximation factor for the \textsc{Latency} problem.  
	In restricted metrics such as trees, planar graphs, and Euclidean plane, \textcite{sitters2021polynomial} gave a PTAS and also showed that even on trees, the problem is $\np$-hard.
	
	A generalization of the \textsc{Latency} problem to a multi-depot setting also has been studied, and constant factor approximation algorithms in this setting are known \cite{fakcharoenphol2007k,chaudhuri2003paths,lp_undirected}.
	
	Another related problem is \textsc{Orienteering}. In this problem, we are given rewards associated with clients located in a symmetric metric space, a length bound $B$, a start (and
	possibly end) location for the vehicle, and we seek a route of length at most $B$ that gathers maximum reward. This problem and its related variants (e.g., the \textsc{$k$-Stroll} problem) often arise as a subroutine in approximation algorithms design for the \textsc{Latency} problem \cite{blum1994minimum,fakcharoenphol2007k,lp_undirected}. 
	Constant factor approximation algorithms for \textsc{Orienteering} are known. \cite{blum2007approximation} presented the first such algorithm with approximation factor $4$, culminating in a $(2+\eps)$-approximation algorithm by \textcite{chekuri2012improved}. All these works are based on combinatorial techniques. \textcite{friggstad2017compact} introduced a natural bidirected LP relaxation for this problem, and presented a rounding algorithm that bounded the integrality gap of this LP by $3$.
	
	As expected, the directed version of this problem called \textsc{Directed Orienteering} has been studied in the literature as well; however, with much less success. In particular, we do not know of any constant-factor approximation for this problem and only a logarithmic approximation is known \cite{nagarajan2011directed,chekuri2012improved}.
	
	\subsection{Organization of the Paper}
	
	We start by presenting the high-level idea of our algorithm in \cref{sec:outline}.
	Besides discussing our main techniques, this section contains important notation that we use later on.
	
	Then, in \cref{sec: lp}, we present the refined linear program we use, which is based on guesses about the structure of an optimal solution.
	The proof of \cref{thm:main} now splits into two parts. 
	
	In \cref{sec: main algorithm}, we describe how we set up this LP, and show that it admits a cheap feasible solution.
	
	Finally, we show how to round any feasible fractional solution.
	We start with discussing basic tools in \cref{sec: preliminary}.
	Then, we present our rounding algorithm in \cref{sec: rounding}.
	
	\section{Our approach}\label{sec:outline} 
	
	In this section, we discuss the main ideas and new techniques we use to prove \cref{thm:main}.
	We also introduce important notation we use later for its formal proof.
	
	We start with introducing a variant of \textsc{Directed Latency} that will simplify the presentation of our results.
	In \textsc{Directed Latency with Target}, we are given an asymmetric metric space $(V \cup \{s, \target\},c)$. The task is to compute an $s$-$\target$-path $P$ so that its total latency $\lat(P) = \sum_{v \in V} c_P(v)$ is minimized. 
	
	Let $\epsilon > 0$.
	Let $(V \cup \{s,\target\},c)$ be an instance of \textsc{Directed Latency with Target}. We can assume $n:=|V|=2^k-1$ for some $k\in\mathbb{Z}$ by colocating an appropriate number of vertices at $s$. Then, by applying scaling techniques as in Theorem 1 in \cite{quasi-poly} we can assume that all distances are positive integers, and that $T := n^2  \cdot \max_{u,v \in V \cup \{s,\target\}} c(u,v)$ is bounded by a polynomial in $n$ and $\frac{1}{\epsilon}$. 
	Note that $T$ constitutes an upper bound on the optimum solution value. We call such an instance of \textsc{Directed Latency with Target} {\em $\epsilon$-nice}.
	
	\begin{definition}\label{def:nice-instances}
		Let $\epsilon > 0$.
		We call an instance $(V \cup \{s,\target\},c)$ of \textsc{Directed Latency with Target} \emph{$\epsilon$-nice} if $n=2^k-1$ for some $k \in \mathbb{Z}$ and each $c(u,v)$ is a positive integer bounded by $\frac{2n^5}{\eps^2}$.
	\end{definition}
	
	\begin{theorem}\label{thm: reduction to bounded instance}
		For every $\epsilon \in (0,1]$ and every $\alpha \in \left[1,\frac{1}{\epsilon}\right]$, if there is an $\alpha$-approximation for $\epsilon$-nice instances of \textsc{Directed Latency with Target}, then there is an $\alpha(1+\eps)$-approximation for \textsc{Directed Latency}.
	\end{theorem}
	
	The proof is essentially the same as the proof of Theorem 1 in \cite{quasi-poly} and is postponed to Appendix \ref{Appendxi A}.
	In this paper, we will show the following.
	
	\begin{theorem}\label{thm:main_nice}
		For every $\epsilon >0$, there is an $O(1)$-approximation for $\epsilon$-nice instances of \textsc{Directed Latency with Target} with running time polynomial in $n$ and $\frac{1}{\epsilon}$.\footnote{We note that only the running time but not the approximation factor depends on $\eps$.}
	\end{theorem}
	
	Note that our main result, \cref{thm:main}, is a direct consequence of \cref{thm: reduction to bounded instance} and \cref{thm:main_nice}.
	We remark that we opted for a cleaner presentation rather than fully optimizing the constant factor, which is in the order of $8 \cdot 10^7$. 
	From now on, we will fix an $\epsilon>0$ for the entire paper and refer to $\epsilon$-nice instances simply as nice instances for ease of notation. 
	
	Before discussing our techniques in more detail, we briefly recap the quasi-polynomial time constant-factor approximation by \textcite{quasi-poly}, both to draw parallels with prior work and as reference to highlight our new ingredients.
	
	\subsection{A constant-factor approximation in quasi-polynomial time}\label{sec:quasi-poly}
	
	We start with introducing the linear programming relaxation on which both, the algorithm in \cite{quasi-poly} as well as our approach is based on. 
	Let $(V \cup \{s,\target\},c)$ be a nice instance of \textsc{Directed Latency with Target}.
	In order to model the objective, we work on the following time-expanded graph.
	For $h \in \mathbb{Z}_{>0}$, we use $[h] := \{1,\dots,h\}$.
	
	\begin{definition}\label{def: time-index graph}
		Define the \emph{time-indexed graph} for $(V \cup \{s,\target\},c)$ to be the graph $G_T$ on vertex set $(V \times [T]) \cup \{(s,0),(\target,T+1)\}$ that contains an edge between $(u,t)$ and $(v,t')$ if and only if $c(u,v) \le t'-t$. 
		For $(v,t) \in V[G_T]$, we abbreviate  $\delta^+(v,t):= \delta^+_{G_T}(\{(v,t)\})$ and $\delta^-(v,t):= \delta^-_{G_T}(\{(v,t)\})$, i.e., the edges going out and going in $(v,t)$, respectively.\footnote{For a directed graph $G=(V,E)$ and $S \subseteq V$, we write $\delta_G^+(S) := \{ (v,w) \in E : v \in S \text{ and } w \notin S\}$. Analogously, we write $\delta_G^-(S) := \{(v,w) \in E :v \notin S \text{ and } w \in S \}$. Moreover, we use $V[G]:= V$ and $E[G]:= E$.}
		Moreover, we write $\delta(v,t) := \delta^+(v,t) \cup \delta^-(v,t)$. 
		For $S \subseteq V[G_T]$, we abbreviate
		$\delta^-(S):= \delta_{G_T}^-(S)$ and $\delta^+(S):= \delta_{G_T}^+(S)$.
	\end{definition}
	
	The feasible region of the LP relaxation we consider is given by the following polytope.
	For $v \in V$ and $t \in [T]$, the variable $x_{v,t}$ indicates that $v$ is visited at time $t$. 
	For $e = ((u,t),(v,t')) \in E[G_T]$, $z_{e}$ indicates that we traverse the edge $(u,v)$ starting at time $t$.
	We remark that by the definition of nice instances, i.e., $c$ is strictly positive, no client can be visited at time $0$.
	
	\begin{equation*}\label{eq:time-indexed-polytope}
		\begin{array}{rrclll}
			Q(G_T) := \Big\{ &(x,z) \in [0,1]^{V[G_T]} \times [0,1]^{E[G_T]} &: \\
			&\displaystyle \sum_{t \in  [T]} x_{v,t} & = & 1 & \forall v\in V \\
			&  \phantom{\displaystyle\sum_{a}} z(\delta^+(v,t)) = z(\delta^-(v,t)) & = &  x_{v, t} & \forall v \in V, t \in [T] \\
			& \phantom{\displaystyle\sum_{a}}
			z(\delta^+(s,0)) = z(\delta^-(\target,T+1)) & = & 1 \\
			&\displaystyle z(\delta^-(S)) & \ge & \displaystyle \sum_{t' \in [t]} x_{v,t'} & \forall v \in V, t \in [T], S \subseteq V[G_T]: \\ &&&&
			(s,0) \notin S,
			\{v\} \times [t]  \subseteq S &\Big\}
		\end{array}
	\end{equation*}
	Note that integer points in $Q(G_T)$ correspond to $s$-$\target$-paths visiting all vertices in $V$.
	Minimizing the total latency of such a path can now be formulated using the $x$-variables.
	\begin{equation}\tag{time-indexed LP relaxation}\label{eq:time-indexed-LP}
		\min \displaystyle \sum_{v\in V, t \in [T]} t x_{v,t} \qquad \textnormal{ s.t. } \qquad (x,z) \in Q(G_T) \enspace.
	\end{equation}
	
	Note that this LP relaxation can be solved in polynomial time as the separation oracle boils down to a polynomial number of min-cut computations.
	For completeness, we provide a short proof  in \cref{appendixB}.
	
	To prove \cref{thm:main_nice}, we will not solve this LP relaxation directly, but first guess structural properties of a fixed optimum solution denoted by $\OPT$ to strengthen the LP. 
	This follows the idea of \textcites{quasi-poly}, who provide a constant-factor approximation for \textsc{Directed Latency} in time $n^{O(\log n)}$.
	In the following, we shortly recap the high-level idea of their algorithm.
	
	Let $h := \lceil \log T \rceil$.
	For $i \in [h] \cup \{0\}$, choose the \emph{bucket threshold} $t_i := 2^i$. 
	For each $i \in [h]$, they guess the last vertex $v^*_i \in V$ visited by $\OPT$ within the time interval $ I_i := [t_{i-1},t_{i})$ together with its exact latency $\ell_i := c_{\OPT}(v)$ in $\OPT$, or guess that no such vertex exists.
	For each such pair $(v^*_i,\ell_i)$, they add the constraints
	\begin{equation}\label{eq:guessed_lat}
		x_{v^*_i,\ell_i} = 1 \qquad \text{ and } \qquad x_{v,t} = 0 \quad \forall v \in V, t \in (\ell_i,t_i) \enspace.
	\end{equation} 
	Then, they solve the resulting LP and obtain a fractional solution $(x,z)$.
	For $v \in V$, define $t(v)$ to be the minimum time so that $\sum_{t \in [t(v)]} x_{v,t} \ge \rho$ for some parameter $\rho \in (\frac12, 1)$.
	They group the vertices into buckets $B_i := \{v \in V : t(v) \in I_i \}$ for $i \in [h]$.
	On each bucket, they compute a path $P_i$ on $B_i$ ending in $v^*_i$ of length at most $\gamma_\rho t_i$ for some constant $\gamma_\rho > 0$ depending on $\rho$, leveraging known LP-relative constant-factor approximation algorithms for \textsc{Asymmetric Path TSP}.  
	We discuss this in more detail in \cref{sec: preliminary}.
	
	Finally, they concatenate these paths to obtain their final solution $P$. 
	They exploit that for two consecutive indices $i$ and $j$ with non-empty buckets, the edge joining $v^*_i$ (the endpoint of $P_i$) and the starting point $u_{j}$ of $P_{j}$ has length at most $t_{j}$. Indeed, the fractional flow $z$ sends some flow from $(v^*_i,\ell_i)$ to $(u_{j},t(u_{j}))$ in $G_T$, using that $x_{v^*_i,\ell_i} = 1$ and $x_{u_{j}, t(u_j)} > 0$.
	Thus, they can bound the total latency of $P$ by 
	\begin{align*}
		\sum_{v \in V} c_P(v) &\le \sum_{i \in [h]} |B_i| \cdot \sum_{j \in [i]} \left(t_j + c(P_j) \right) \\
		&\le \sum_{i \in [h]} |B_i| \cdot \sum_{j \in [i]} \left(1+\gamma_\rho \right) 2^j \\
		&< 4\left(1+\gamma_\rho \right) \sum_{i \in [h]} |B_i| \cdot 2^{i-1} \\ 
		&\le \frac{4}{1-\rho}\left(1+\gamma_\rho \right) \sum_{v \in V, t \in [T]} t x_{v,t} \\
		&\le \frac{4}{1-\rho}\left(1+\gamma_\rho \right) \lat(\OPT) \enspace,
	\end{align*}
	where we used in the first inequality that we can bound the latency of a vertex in bucket $B_i$ by $\sum_{j \in [i]} \left(t_j + c(P_j) \right)$ by construction, and in the second last inequality that for each $v \in B_i$, we have $\sum_{t \in [T]} tx_{v,t} \ge (1-\rho) 2^{i-1}$ by definition of $B_i$.
	This yields the desired constant-factor approximation.
	
	To achieve polynomial running time, we cannot afford to guess the vertices $v^*_i$ and their exact latency $\ell_i$ for all these time intervals.
	Without knowing these vertices, we could still compute for each $i \in [h]$ a path $P_i$ of length at most $\gamma_\rho t_i$ visiting all vertices in $B_i$.
	However, concatenating these paths could be very expensive since the fractional flow $z$ might send no flow from the endpoint of $P_i$ to the starting point $u_j$ of $P_j$ up to time $t(u_j)$.
	
	In the remainder of this section we introduce all novel ingredients needed for our constant-factor approximation algorithm.
	We start with discussing a desirable property so that, if this property holds, we are still able to bound the cost resulting from concatenating the paths $P_1,\dots,P_h$.
	
	
	\subsection{A desirable structural property}\label{sec:desirable-property}
	
	Let $B_1,\dots,B_h$ be buckets defined as in the previous subsection,
	and let $P_i$ be a path visiting all vertices in $B_i$ for each $i \in [h]$. For the sake of keeping this exposition short, we assume that for each $i\in [h]$, we have $c_{\OPT}(v)\in I_i$ for some $v\in V$. In \cref{sec:high-level-roots}, we discuss that this assumption is guaranteed by our new way of defining the bucket thresholds. 
	
	For $i \in [h-1]$, we show how to bound the cost resulting from joining the endpoint $u$ of $P_i$ with the starting point $v$ of $P_{i+1}$.
	As already discussed, we can conclude that $c(u,v) < t_{i+1}$ if $z$ sends some flow from $u$ to $v$ up to time $t(v)$.
	Otherwise, using $\rho > \frac12$, $z$ sends some flow from $v$ to $u$ up to time $t(u)$, yielding  $c(v,u) < t_{i+1}$.
	Thus, if 
	\begin{equation}\label{eq:desirable_prop}
		c(u, v) \le \max\{c(v,u),t_{i+1}\} \enspace,
	\end{equation}
	we could still conclude that $c(u,v)$ is small, namely 
	\begin{equation*}
		c(u,v) \le \max\{\min\{c(u,v), c(v,u)\}, t_{i+1}\} \le t_{i+1} \enspace.
	\end{equation*}
	If \eqref{eq:desirable_prop} holds, we call the edge $(u,v)$ \emph{$t_{i+1}$-short}.
	
	We aim for adding constraints to the \ref{eq:time-indexed-LP} so that for buckets $B_1,\dots,B_h$ induced by an optimal fractional solution to this strengthened LP, we have that
	\begin{equation}\label{eq:desired}
		(u,v) \text{ is $t_{i+1}$-short  for each $u \in B_i$ and $v \in B_{i+1}$ with $i \in [h-1]$.}
	\end{equation}
	
	One way to achieve this is by adding the following constraints to the \ref{eq:time-indexed-LP}:
	\begin{equation}\label{eq: new const without tours}
		\sum_{t \in [t_{i}-1]} x_{u,t} \le \sum_{t \in [t_{i}-1]} x_{v,t} \quad \forall u,v \in V, i \in [h-1]: (u,v) \text{ is not $t_{i+1}$-short} \enspace.
	\end{equation}
	
	Why does \eqref{eq: new const without tours} imply \eqref{eq:desired}? Let $i \in [h-1]$, $u\in B_i$ and $v\in B_{i+1}$. If $(u,v)$ is $t_{i+1}$-short then there is nothing to prove. Otherwise, by \eqref{eq: new const without tours}, $v$ gets at least $\rho$ units of flow before time $t_{i}-1$, a contradiction with the fact that $v\in B_{i+1}$.
	
	Let $(x,z)$ be the solution to the \ref{eq:time-indexed-LP} induced by $\OPT$.
	It is easy to see that \eqref{eq: new const without tours} is satisfied for vertices $u$ and $v$ whenever $v$ is visited before $u$ in $\OPT$. So let $u$ be visited before $v$.
	It is easy to see that if $c_{\OPT}(u)\in I_j$ and $c_{\OPT}(v)\in I_j\cup I_{j+1}$ for some $j \in [h-1]$, then the pair $u,v$ indeed satisfies \eqref{eq: new const without tours}. More precisely, for $i=j$ we have that $(u,v)$ is $t_{i+1}$-short, and for $i<j$ or $i>j$ the inequality in \eqref{eq: new const without tours} trivially holds. 
	
	However, in general $\OPT$ does not need to satisfy \eqref{eq: new const without tours}.
	In particular, it does not need to hold if $c_{\OPT}(u)\in I_i$ and $c_{\OPT}(v)\in I_{j+1}$ with $j > i$, and $c(v,u)<c(u,v)$.
	However, we observe that in this case, the $u$-$v$-path in $\OPT$ together with the edge $(v,u)$ imply a short cycle (to which we refer as tours) of length at most $2 c_{\OPT}(v)\leq 2 t_{j+1}$ on all vertices in $\OPT$ visited within the time-interval  $I_j$, see \cref{fig: intro backward edge}.
	
	Thus, one can show that by slightly relaxing the bucket thresholds $t_1,\dots,t_h$, there is an index set $A_{\textnormal{tour}} \subseteq [h]$ and a path $P$ in $G_T$ (which results from traversing $\OPT$, but delaying it in certain situations, see \cref{fig: intro bucket expansion}) with the following properties. 
	(Refer to \cref{fig: intro everything} for intuition.)
	For $v \in V$, let $\lat_P(v)$\footnote{We use $\lat_P(v)$ to emphasis that $P$ is in $G_T$ and we use $c_{P}(v)$ when $P$ is a path in $G$.}denote the time we visit $v$ in $P$.
	\begin{enumerate}[label=(\roman*)]
		\item\label{item:short_tour} for each $i \in A_{\textnormal{tour}}$, there is a short cycle of length at most $2t_i$ on $\{v \in V: \lat_P(v) \in I_i\}$ (recall that we abbreviate $I_i = [t_{i-1},t_i)$), and \item\label{item:ordering} for each $i,j \in [h] \setminus A_{\textnormal{tour}}$ with $i < j$, and each $u,v \in V$ with $\lat_P(u) \in I_i$ and $\lat_P(v) \in I_j$, the edge $(u,v)$ is $t_{i+1}$-short.
	\end{enumerate}
	Note that there are at most $2^h$ many options for $A_\textnormal{tour} \subseteq [h]$ which is polynomial in $n$ and $\frac{1}{\eps}$. Thus, we can efficiently enumerate all possible choices.
	We call intervals $I_i$ with $i \in A_{\textnormal{tour}}$  \emph{tour-intervals}.
	In words, \ref{item:ordering} guarantees that for $u,v \in V$ with $u$ being visited by $\OPT$ in some earlier interval $I_i$ than $v$ either
	\begin{itemize}
		\item $(u,v)$ is $t_{i+1}$-short, or
		\item $u$ or $v$ is visited in some tour-interval.
	\end{itemize}
	
	If we would guess for each $i \in A_{\textnormal{tour}}$ one vertex $v^*_i \in \{v \in V: c_P(v) \in I_i\}$, we could strengthen the \ref{eq:time-indexed-LP} by enforcing entering and leaving the tour-interval $I_i$ via $v^*_i$. We remark that this requires to slightly relax the bucket thresholds to maintain feasibility, see \cref{fig: intro tour interval} for an illustration of this.
	Note that enforcing entering and leaving the tour-interval $I_i$ via $v^*_i$ can be achieved by forbidding all other edges crossing the tour-interval as elaborated in \cref{sec: lp}.  
	
	\begin{figure}
		\centering
		\begin{subfigure}{0.9\textwidth}
			\includegraphics[scale=0.9]{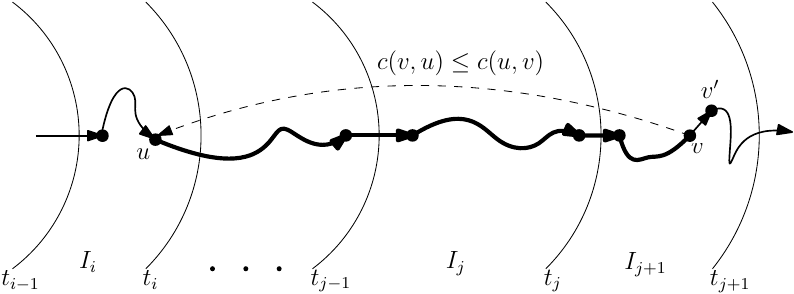}
			\caption{}
			\label{fig: intro backward edge}
		\end{subfigure}
		\begin{subfigure}{0.9\textwidth}
			\includegraphics[scale=0.9]{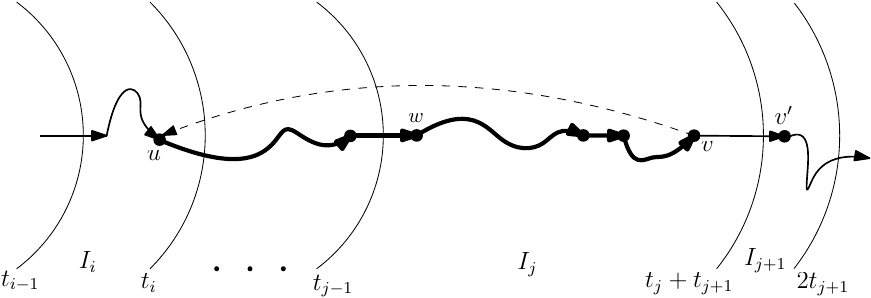}
			\caption{}
			\label{fig: intro bucket expansion}
		\end{subfigure}
		\begin{subfigure}{0.9\textwidth}
			\centering
			\includegraphics[scale=0.9]{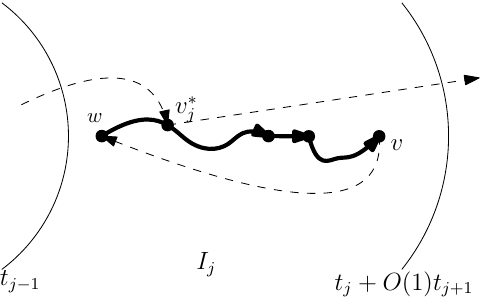}
			\caption{}
			\label{fig: intro tour interval}
		\end{subfigure}
		\caption{(a) We have a pair of vertices $u,v$ with $c_{\OPT}(u)\in I_i$,  $c_{\OPT}(v)\in I_{j+1}$, $i<j$, and $c(v,u) < c(u,v)$. 
			Assume that $(u,v)$ is not $t_{i+1}$-short, so we cannot impose the constraints in \eqref{eq: new const without tours}. 
			Note that $(v,u)$ together with the thick sub-path of $\OPT$ from $u$ to $v$ creates a cycle of length at most $2 c_{\OPT}(v)\leq 2 t_{j+1}$. (b) We relax the time interval $I_j$ so that $c_{\OPT}(v)\in I_{j}$ and we add $j$ to $A_{\textnormal{tour}}$. 
			Then, we can visit $v$ in the new tour-interval $I_j$, which is essential to guarantee \ref{item:ordering}.
			Let $v'$ be the vertex visited directly after $v$ in $\OPT$. We choose $((v,t),(v',t')) \in E[G_T]$ joining $v$ and $v'$ (perhaps with $t'-t > c(v,v')$) that delays the visit of vertices coming after $v$ in $\OPT$ so they are visited in the next time interval. (c) Let $v^*_j$ be any vertex visited by $\OPT$ within $I_j$. By adding the dashed edges, we enter and leave $I_j$ via $v^*_j$. It is easy to see that the cost of all the added edges can be bounded by $O(1) t_{j+1}$. Thus, by further relaxing the time interval $I_j$, we can guarantee to enter and leave $I_j$ via $v^*_j$ and visit all vertices on the $w$-$v$-path in $\OPT$ in between, where $w$ is the first vertex visited by $\OPT$ within $I_j$.}
		\label{fig: intro everything}
	\end{figure}
	
	Additionally, we add the following constraints capturing \ref{item:ordering}, which are similar to those in \eqref{eq: new const without tours}:
	\begin{equation}\label{eq:ordered-lp-constraint2} 
		\sum_{t \in [t_{i}-1]} (x_{u,t} - x_{v,t}) \le \sum_{t \in I_\textnormal{tour}} (x_{u,t} +x_{v,t} ) \quad \forall u,v \in V, i \in [h-1]: (u,v) \text{ is not $t_{i+1}$-short}\enspace,
	\end{equation}
	where $I_\textnormal{tour} := \bigcup_{i \in A_\textnormal{tour}} I_i$.
	As discussed in the following subsection, these constraints indeed allow us to round a fractional solution while losing only a constant factor.
	
	Still, for each $i \in A_{\textnormal{tour}}$, there are $n$ many options to choose $v^*_i$.
	With a new choice of the bucket thresholds $t_i$, and using \ref{item:short_tour}, we show that we can efficiently compute for each $i \in A_{\textnormal{tour}}$ a vertex $\hat{v}_i$ that is close to any vertex in $\{v \in V: c_{\OPT}(v) \in I_i\}$, which suffices for our purposes.
	In the LP, we then enforce to enter and leave the tour-interval $I_i$ via $\hat{v}_i$.  
	We discuss this in more detail in \cref{sec:high-level-roots}.
	
	\subsection{Rounding fractional solutions}\label{sec:rounding-high-level}
	
	Let $(x,z) \in Q(G_T)$ satisfying \eqref{eq:ordered-lp-constraint2} so that each tour-interval $I_i$ is entered and left via $\hat{v}_i$.
	On a high-level, we round $(x,z)$ following the idea in \cite{quasi-poly}. 
	More precisely, we put the vertices into buckets, compute a path on each bucket, and then concatenate these paths. 
	On the $i$'th bucket we compute a path of length $O(t_i)$.
	Moreover, concatenating the path on bucket $i-1$ with the path on bucket $i$ will cost at most $O(t_i)$. 
	Since the bucket thresholds $t_0,\dots,t_h$ are exponentially increasing, this guarantees that in the resulting path, the latency for each vertex in the $i$'th bucket is in $O(t_i)$.
	
	The main new ingredient is to distinguish between vertices $V_\textnormal{tour} := \{v \in V : \sum_{t \in I_\textnormal{tour}} x_{v,t} \ge \delta\}$ which are visited to at least a small extent of $\delta >0$ within tour-intervals, and those which are visited to an extent of at least $1-\delta$ outside tour-intervals.
	
	Building on the definition in \cref{sec:quasi-poly}, let $t_\textnormal{tour}(v)$ for $v \in V_\textnormal{tour}$ be the minimum time so that $\sum_{t \in  [t_\textnormal{tour}(v)] \cap I_\textnormal{tour}} x_{v,t} \ge \rho_1 \delta$ for some parameter $\rho_1 \in (0,1)$. 
	Moreover, for $v \in V \setminus  V_\textnormal{tour}$, define $t_\textnormal{non-tour}(v)$ to be the minimum time so that $\sum_{t \in [t_\textnormal{non-tour}(v)] \setminus I_\textnormal{tour}} x_{v,t} \ge \rho_2 (1-\delta)$ for some parameter $\rho_2 \in (\frac12,1)$. 
	
	For tour-intervals, we exploit that for every vertex $v$ which is visited within a tour-interval $I_i$ to some fractional extent $\delta'$, $z$ sends a flow of $\delta'$ from $\hat{v}_i$ to $v$ and vice versa.
	Thus, using a simple scaling argument and standard tools from graph theory, we can for each tour-interval $I_i$ construct a feasible solution to the natural LP relaxation for \textsc{Asymmetric Path TSP} on vertex set $\{s,\hat{v}_i\} \cup B_i$ with starting point $s$ and endpoint $\hat{v}_i$, where 
	\begin{equation*}
		B_i := \{v \in V_\textnormal{tour} : t_\textnormal{tour}(v) \in I_i \} \enspace.
	\end{equation*}
	We then use known LP-relative algorithms for \textsc{Asymmetric Path TSP} to round this solution.
	Note that the union of these paths covers all vertices in $V_\textnormal{tour}$. 
	
	For the remaining time intervals $I_i$ with $i \in [h] \setminus A_\textnormal{tour}$, we introduce a new definition of buckets, leveraging our new constraints stated in \eqref{eq:ordered-lp-constraint2}.
	This new definition will enable us to bound the extra cost resulting from concatenating the paths on each bucket.
	We start with preliminary buckets 
	\begin{equation*}
		B_i := \{v \in V \setminus V_\textnormal{tour} : t_\textnormal{non-tour}(v) \in I_i\} \enspace.
	\end{equation*}
	In the following, we assume that $i+1 \in [h] \setminus A_\textnormal{tour}$.
	(This is arguably the most interesting case, illustrating the use of the constraints in \eqref{eq:ordered-lp-constraint2}. Of course, in our proofs we need to handle the other cases as well.)
	Let $W_i$ be the set of vertices $v \in V \setminus V_{\textnormal{tour}}$ for which there is some $u \in B_i$ so that $(u,v)$ is not $t_{i+1}$-short.
	Thanks to our new constraints stated in \eqref{eq:ordered-lp-constraint2}, we know that in this case, $v$ is visited to an extend of at least $\rho_2(1-\delta) - 2\delta$ up to time $t_i$.
	Indeed, $u$ is visited to an extend of at least $\rho_2(1-\delta)$ up to time $t_i$, and both $u$ and $v$ are visited to an extend of at most $\delta$ within tour-intervals.
	We choose $\delta$ and $\rho_2$ so that $\rho_2(1-\delta) - 2\delta > \frac12$, guaranteeing the existence of a short path on $B_i \cup W_i$ by using similar arguments as in \cite{quasi-poly}.
	Thus, we add $W_i$ to $B_i$ and update the other buckets accordingly. 
	Denote the resulting buckets by $\overline{B}_1,\dots,\overline{B}_h$.
	We now have that $(u,v)$ is $t_{i+1}$-short for each $u \in \overline{B}_i$ and $v \in \overline{B}_{i+1}$ with $i,i+1 \in [h] \setminus A_\textnormal{tour}$.
	As discussed in the previous subsection, this allows to bound the extra cost resulting from concatenating the path on $\overline{B}_i$ with the path on $\overline{B}_{i+1}$.
	
	On each of the final buckets, we compute a path $P_i$.
	It remains to argue why the extra cost resulting from concatenating $P_i$ and $P_{i+1}$ is small if $I_i$ or $I_{i+1}$ is a tour-interval. 
	We will show this by exploiting that we enforced to enter and leave each tour-interval $I_j$ via $\hat{v}_j$.
	For example, if $I_{i+1}$ is a tour-interval, then $z$ sends some flow from any vertex $u$ in $P_i$ via $\hat{v}_{i+1}$ to any vertex $v$ in $P_{i+1}$. This yields $c(u,v) \le t_{i+1}$.
	The case that $I_i$ is a tour-interval is more involved and actually needs some further adaptation of the buckets, as we discuss in \cref{sec: rounding}.
	On a high-level, the issue is that for $v \in \overline{B}_i \setminus B_i$, we do not have the guarantee that $\sum_{t \in I_i} x_{v,t} > 0$.
	
	\subsection{Computing roots for tour intervals: A new choice of bucket thresholds}\label{sec:high-level-roots}
	
	We finally argue how to efficiently compute for each tour-interval $I_i$ a vertex $\hat{v}_i \in V$ that is close to any vertex in $\{v \in V: c_{\OPT}(v) \in I_i\}$. Recall that $I_i = [t_{i-1},t_i)$.
	Note that initially we set $t_i := 2^i$ for $i \in [h] \cup \{0\}$.
	This choice allowed to conclude that if we pay latency $t_i$ for each vertex $v \in V$ with $c_{\OPT}(v) \in I_i$, then we end up with a total latency of at most $2\lat(\OPT)$.
	
	Another way to get a similar guarantee is to choose bucket thresholds $t_1,\dots,t_k$ so that $\OPT$ visits exactly $\frac{n+1}{2^i}$ vertices within $I_i = [t_{i-1},t_i)$ for each $i \in [k]$, where $k := \log( n+1 )$.
	(Note that $k$ is an integer by definition of nice instances.)
	Indeed, if we pay latency $t_{i}$ for each vertex $v \in V$ with $c_{\OPT}(v) \in I_i$, we end up with a total latency of at most
	\begin{align*}
		\sum_{i = 1}^k \frac{n+1}{2^i} t_{i} \le 2 \sum_{i = 1}^k \frac{n+1}{2^i} t_{i-1} + t_k  \le 3 \lat(\OPT) \enspace.
	\end{align*}
	Clearly, we do not know these values of $t_i$, but we can efficiently guess $t_i$ rounded to the next power of $2$ for each $i \in [k]$ via enumerating all possible outcomes, while getting similar guarantees.
	Denote these rounded values by $\overline{t}_1,\dots,\overline{t}_k$.
	We will group certain subsequent intervals to avoid that we have many subsequent values of $\overline{t}_i$ which are the same. 
	This grouping guarantees that the bucket thresholds are exponentially increasing, which we crucially exploit when rounding fractional solutions as discussed in the first paragraph of \cref{sec:rounding-high-level}.
	
	Using essentially the same arguments as in \cref{sec:desirable-property}, we can slightly relax these bucket thresholds so that there is a set $A_\textnormal{tour} \subseteq [k]$ satisfying properties  \ref{item:short_tour} and \ref{item:ordering} on page \pageref{item:short_tour}.
	The main advantage of these new bucket thresholds is
	that for each $i \in A_\textnormal{tour}$, we can efficiently compute a vertex $\hat{v}_i$ with $\max\{c(\hat{v}_i, v), c(v,\hat{v}_i)\} \le O(t_i)$ for any $v \in V$ with $c_{\OPT}(v) \in I_i$.
	We formally prove some variant of this statement in \cref{sec: main algorithm}.
	The proof is based on the fact that for any vertex $v^*$ with $c_{\OPT}(v^*) \in I_i$, there are 
	\begin{itemize}
		\item at least $\frac{n+1}{2^i}$ vertices $v \in V$ with $\max\{c(v^*,v), c(v,v^*)\} \le 2 t_i$ (namely those with $c_{\OPT}(v) \in I_i$), and
		\item less than $\frac{n+1}{2^i}$ vertices $v \in V$ with $c(v,v^*) > 2t_i$. (Indeed, this inequality can only hold for vertices $v$ with $c_{\OPT}(v) \ge t_i$).
	\end{itemize}
	One can show that any vertex $\hat{v}_i$ satisfying these two properties must be close to $v^*$.
	More precisely, $\max\{c(\hat{v}_i,v^*), c(v^*,\hat{v}_i)\} \le 4 t_i$. 
	Moreover, we can efficiently compute a vertex $\hat{v}_i$ satisfying these two properties, if it exists, by just checking these properties for every vertex.
	
	For $\hat{v}_i$, we create a copy $r_i$ colocated at $\hat{v}_i$.
	We call $r_i$ the \emph{root} of the tour-interval $I_i$.
	Having these roots at hand, we are finally ready to set up the strengthened version of the \ref{eq:time-indexed-LP}.
	
	\section{Strengthening the \ref{eq:time-indexed-LP}}\label{sec: lp}
	
	In this section, we set up the LP based on our guesses about the structure of an optimal solution.
	This LP refines the \ref{eq:time-indexed-LP} for \textsc{Directed Latency}.
	We are given time intervals $I_i = [t_{i-1},t_i)$ with $t_i \ge \frac{4}{3}t_{i-1}$ for each $i \in [q]$ with $q \le \log (n+1)$.
	Moreover, we are given a set $A_{\textnormal{tour}} \subseteq [q]$ indicating the set of tour-intervals. 
	Previously, we used as upper bound on the total latency $T := n^2 \cdot \max_{u,v \in V \cup \{s,\target\}} c(u,v)$.
	We have to account for the shift in the bucket thresholds as discussed in \cref{sec:outline} and shift $T$ by $t_q$, as we will discuss in detail in \cref{sec: main algorithm}.
	
	Additionally, we are given roots $\roots = (r_i)_{i \in A_{\textnormal{tour}}}$.
	Each root is colocated at some vertex $v \in V$.
	Let $F_{\roots} \subseteq E[G_T]$ be the set of edges $e = ((u,t),(v,t'))$ for which there is some $i \in A_\textnormal{tour}$ so that either 
	\begin{itemize}
		\item $t < t_{i-1} \le t'$ and $v \neq r_i$, or
		\item $t < t_i \le t'$ and $u \neq r_i$.
	\end{itemize}
	In words, $F_{\roots}$ contains for each tour-interval $I_i$ the edges $(u,v)$ entering $I_i$ and the edges $(v,w)$ leaving $I_i$ with $v \neq r_i$, as well as edges skipping $I_i$.
	We add the constraints $z_e = 0$ for each $e \in F_R$ to enforce entering and leaving $I_i$ via $r_i$.
	Recall that we write $I_\textnormal{tour} := \bigcup_{i \in A_\textnormal{tour}} I_i$.
	Now, consider the following LP, which is based on the time-indexed graph $G_T$ for vertex set $V \cup \roots \cup \{s,\target\}$.
	
	\begin{equation}\tag{strengthened time-indexed LP}\label{eq:refined-LP}
		\begin{array}{rrcll}
			\min & \displaystyle \sum_{v\in V, i \in [q], t \in I_i} t_i x_{v,t} \\
			&\phantom{\displaystyle \sum_{a}} (x,z) & \in & Q(G_T) & \\
			& \phantom{\displaystyle \sum_{a}} z(\delta^-(r,t)) & = & z(\delta^+(r,t)) & \forall r \in R, t \in [T] \cup \{0\} \\
			& \phantom{\displaystyle \sum_{a}} z_e & = & 0 & \forall e \in F_{\roots}  \\
			& \displaystyle \sum_{t \in [t_{i}-1]} (x_{u,t} - x_{v,t}) & \le & \displaystyle \sum_{t \in I_\textnormal{tour}} (x_{u,t} +x_{v,t} )& \forall u,v \in V, i \in [q-1]: (u,v) \text{ is not $t_{i+1}$-short}\enspace. 
		\end{array}
	\end{equation}
	
	Note that we can still efficiently compute an optimum solution to this LP as, in comparison to the \ref{eq:time-indexed-LP}, we only added a polynomial number of constraints.
	
	If all our guesses have been correct, then there exists a solution to this linear program that does not cost much more than the total latency of an optimum solution.
	
	\begin{theorem}\label{thm:existence}
		Let $(V \cup \{s,s'\},c)$ be a nice instance of \textsc{Directed Latency with Target} and let $\OPT$ denote an optimum solution. 
		We can efficiently compute a polynomial sized set of triples $(\mathcal{T}=(t_0,\dots,t_q),A_{\textnormal{tour}},R)$ so that for one of these triples, the \ref{eq:refined-LP} admits a solution of cost at most $836 \cdot \lat(\OPT)$.
	\end{theorem}
	
	We prove this theorem in \cref{sec: main algorithm} by constructing a solution based on $\OPT$ so that the additional constraints are satisfied by
	\begin{itemize}
		\item including the roots, using that the resulting detour is small, and
		\item delaying traversing $\OPT$ in certain situations so that we visit each vertex in its desired interval.
	\end{itemize}
	
	Given a fractional solution to the \ref{eq:refined-LP}, we can round this solution while losing only a constant factor on the objective as discussed in \cref{sec:rounding-high-level}. The following theorem is formally proven in \cref{sec: rounding}.
	
	\begin{theorem}\label{thm:rounding}
		Let $(V \cup \{s,s'\},c)$ be an instance of \textsc{Directed Latency with target}. 
		Given a solution $(x,z)$ to the \ref{eq:refined-LP}, we can efficiently compute a solution $P$ with 
		\begin{equation*}
			\sum_{v \in V} c_P(v) \le 96473 \cdot \sum_{v\in V, i \in [q], t \in I_i} t_i x_{v,t} \enspace.
		\end{equation*}
	\end{theorem}
	
	Note that \cref{thm:existence,thm:rounding} directly yield \cref{thm:main_nice}, namely a constant-factor approximation for nice instances of \textsc{Directed Latency with Target}.
	Recall that this yields the desired constant-factor approximation for \textsc{Directed Latency} by using \cref{thm: reduction to bounded instance}.
	
	\section{Existence of a good LP solution (Proof of \cref{thm:existence})}\label{sec: main algorithm}
	
	As discussed in \cref{sec:outline}, our strengthened LP is based on notion of bucket thresholds, tour-intervals, and roots. We define these notions precisely in \cref{subsec: setting up the lp}, and provide some further intuition behind these definitions.
	We then present the algorithm that determines the bucket thresholds, tour-intervals and roots needed to set up the \ref{eq:refined-LP} at the end of \cref{subsec: setting up the lp}.
	More precisely, our approach includes guesses on the structure of an optimum solution, which we implement by enumerating all possible outcomes. 
	We will guarantee that the number of execution paths created this way is bounded by a polynomial in $n$.
	
	As discussed in the introduction, the resulting LP, i.e. \ref{eq:refined-LP}, is no longer a relaxation of \textsc{Directed Latency}.
	We can still show that if our guesses about the structure of an optimum solution have been correct, it admits a solution of cost at most $O(1)$ times the optimum total latency as stated in \cref{thm:existence}.This is proven in \cref{sub: cheap lp sol}.
	
	In \cref{sec: rounding} we then show how to round any feasible fractional solution to a solution of \textsc{Directed Latency With Target} whose latency is at most a constant factor times the LP-value of the fractional solution.   
	
	\subsection{Setting up the \ref{eq:refined-LP}}\label{subsec: setting up the lp}
	
	Let $(V \cup \{s,\target\},c)$ be a nice instance of \textsc{Directed Latency with Target}.
	Throughout this section, we fix an optimal solution $\OPT$ with total latency $\opt$. Let $s,v^*_1,v^*_2,\ldots,v^*_{n-1},v^*_n,\target$ be the sequence of vertices as visited by $\OPT$.
	We start by splitting $\OPT$ into certain {\em buckets} so that each bucket contains half as many vertices than the previous one. 
	Then, we {\em group} these buckets based on their distances from $s$. We show that there are polynomially many such groups and after guessing which one is a tour-group as discussed in \cref{sec:outline}, we compute a root vertex for each tour-group.
	After slightly modifying these groups, we then define bucket thresholds based on these groups. We use these bucket thresholds, tour-groups, and their roots to set up the \ref{eq:refined-LP}.
	
	
	
	Let $B^*_1 \subseteq V$ denote the first $\frac{n+1}{2}$ vertices visited by $\OPT$, $B^*_2 \subseteq V \setminus B^*_1$ the next $\frac{n+1}{4}$ vertices, and so on.
	More formally, $B^*_1:=\{v^*_1,v^*_2,\ldots,v^*_{\frac{n+1}{2}}\}$, $B^*_2:=\{v^*_{\frac{n+1}{2}+1},\ldots,v^*_{\frac{n+1}{2}+\frac{n+1}{4}}\}$, $\ldots$, and $B^*_{k}:=\{v^*_{n}\}$, where $k= \log (n+1)$. Each $B^*_i$ is called a {\em bucket}.
	
	For every $B^*_i$, we denote by $\maxlat(B^*_i)$ the {\em maximum latency} of $B^*_i$ which is defined to be the largest latency w.r.t. $\OPT$ among all vertices in $B^*_i$.
	
	We would like to guess $\maxlat(B^*_i)$ for each $i \in [k]$. Unfortunately, the total number of options is quasi-polynomial in $n$. 
	As we will discuss soon, we can reduce the number of guesses by rounding each $\maxlat(B^*_i)$ to the next power of $2$, denoted by $\maxlatrounded(B^*_i)$.
	
	
	For $i \in [k-2]$, the indices $i$ and $i+1$ belong to the same \emph{group} if and only if $\maxlatrounded(B^*_{i}) = \maxlatrounded(B^*_{i+1}) = \maxlatrounded(B^*_{i+2})$.
	Let $\group^*_1,\dots,\group^*_{q^*}$ denote the resulting groups. Let $V(\group^*_i):=\bigcup_{j\in\group^*_i}B^*_j$, and let $s(\group^*_i):=|V(\group^*_i)|$. See Figure \ref{fig: groups} for an illustration of buckets and groups.
	
	\begin{figure}
		\centering
		\includegraphics[scale=0.9]{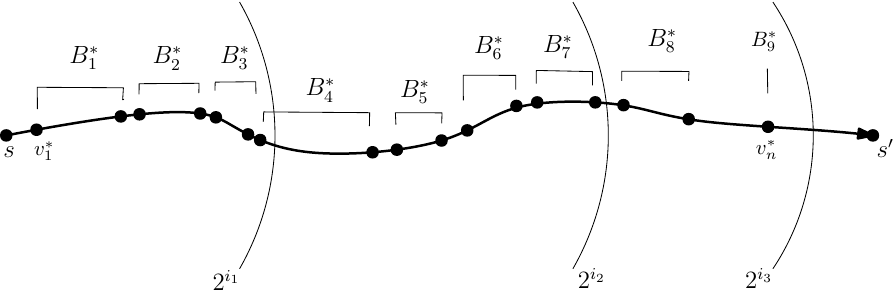}
		\caption{In this example, we have $n=2^9-1$. The shown $s,s'$-path is $\OPT$. The buckets are indicated with their first and last vertices. The number besides the semi-circles show the distances from $s$. In this example, we then have $\group^*_1:=\{1,2\}, \group^*_2:=\{3\}, \group^*_3:=\{4,5,6\}, \group^*_4:=\{7\}, \group^*_5:=\{8\}, \group^*_6:=\{9\}$. Also, we have $\maxlatrounded(\group^*_1)=\maxlatrounded(\group^*_2)=2^{i_1}$, $\maxlatrounded(\group^*_3)=\maxlatrounded(\group^*_4)=2^{i_2}$, $\maxlatrounded(\group^*_5)=\maxlatrounded(\group^*_6)=2^{i_3}$.}
		\label{fig: groups}
	\end{figure}
	
	Note that by the above definition, if multiple buckets have the same rounded latency, then we group all but the last one in one group and the last bucket forms a separate group. Doing so, we have the property that for $i\in [q^*]$ and $j\in \group^*_i$ we have $\maxlatrounded(B^*_{j+1})=\maxlatrounded(\group^*_{i+1})$ which will be used in \cref{sub: cheap lp sol} in bounding the objective value of our LP solution for guesses that are compatible with $\OPT$. Furthermore, the subsequence of groups consisting of every other index has the property that their rounded maximum latency forms a geometric sequence with gap at least $2$ between each consecutive groups in the subsequence.
	This will be used in the analysis of our rounding algorithm. The discussion above motivates the following definition of a valid sequence of groups. Recall $T := n^2  \cdot \max_{u,v \in V \cup \{s,\target\}} c(u,v)$ which is bounded by $\frac{2n^7}{\eps}$ by definition of nice instances.
	
	\begin{definition}\label{def: valid groups}
		Let $1\leq \ell_1\leq\ldots\leq\ell_k\leq 2^{\lceil \log T \rceil}$ be a sequence of $k = \log(n+1)$ values where each $\ell_i$ is a power of $2$. For $i\in[k-2]$, the indices $i$ and $i+1$ belong to the same group if and only if $\ell_i=\ell_{i+1}=\ell_{i+2}$. Let $\group_1,\ldots,\group_q$ denote the resulting groups. Let $\maxlatrounded(\group_i):=\ell_{m_i}$ where $m_i$ is the largest element in $\group_i$, and let $s(\group_i):= \sum_{j \in \group_i} \frac{n+1}{2^j}$. We say $\group_1,\ldots,\group_q$ is a {\em valid sequence of groups}.
	\end{definition}
	
	Note by our construction of groups based on $\OPT$ earlier in this section, $\group^*_1,\ldots\group^*_{q^*}$ is a valid sequence of groups.
	
	\begin{lemma}\label{lem: number of valid groups}
		The number of valid sequences of groups is polynomial in $n$ and $\frac{1}{\eps}$ and we can efficiently enumerate all of them.
	\end{lemma}
	\begin{proof}
		By Definition \ref{def: valid groups}, a valid sequence of groups is obtained uniquely from a sequence of $k$ numbers $1\leq \ell_1\leq\ldots\leq\ell_k\leq 2^{\lceil \log T \rceil}$ where each $\ell_i$ is a power of $2$. Hence, it is sufficient to bound the number of the latter sequences and show how to enumerate all of them. 
		
		Consider a non-negative integer solutions to $x_0+\ldots +x_h=k$ for $h:=\lceil\log T\rceil$. 
		Intuitively, $x_i$ encodes the number of indices $j$ with $\ell_j=2^i$. This gives an injective mapping from sequences of $\ell_1,\ldots,\ell_k$ to solutions of $x_0+\ldots +x_h=k$. It is known that the number of non-negative integer solutions of $x_0+\ldots +x_h=k$ is at most ${k+h\choose k} \leq 2^{k+h+1}$ which is polynomial in $n$ and $\frac{1}{\eps}$. So we can efficiently enumerate all valid choices of $\ell_1,\dots,\ell_k$ in lexicographic order.
	\end{proof}
	
	Let $P^*_i$ be the subpath of $\OPT$ on $V(\group^*_i)$. If we know for each $i \in [q^*]$ one vertex on the subpath $P^*_i$, we could strengthen the \ref{eq:time-indexed-LP} as in \cite{quasi-poly}.
	However, guessing such vertices would result in quasi-polynomial running time.
	We will see that our new buckets will enable us to efficiently compute vertices close to $V(\group^*_i)$ in certain situations. Suppose there is a tour on the vertex set $V(\group^*_i)$ of cost no more than $2\cdot\maxlatrounded(\group^*_{i+1})$ (using $\group^*_{q^*+1}:= \group^*_{q^*}$). Then, (i) for any vertex in $v^*\in V(\group^*_i)$ there are at least $s(\group^*_i)$ vertices, namely those in $V(\group^*_i)$, with distance at most $2\cdot\maxlatrounded(\group^*_{i+1})$ to and from $v^*$. Furthermore, (ii) there are less than $s(\group^*_i)$ vertices with distance larger than $2\cdot \maxlatrounded(\group_i)$ to $v^*$. 
	Indeed, the latter can only hold for vertices in $\bigcup_{j=i+1}^{q^*}V(\group^*_j)$. We efficiently compute a vertex called \emph{root} close to $v^*$ by exploiting that any two vertices with properties (i) and (ii) must be close to each other. This motivates the next definition. We call $\group^*_i$ a \emph{tour-group} if there is a tour on $V[P^*_i]$ of length at most $2\maxlatrounded(\group^*_{i+1})$.
	
	For technical reasons, we choose these roots for a tour-group to not be in $V$, i.e., for each $v\in V$ we make a colocated copy of $v$, denote by $\rootset$ the set of new vertices, and require the set $\roots$ of roots to be a subset of $\rootset$.
	
	\begin{lemma}\label{lem: computing a root}
		Let $\group_1,\ldots,\group_q$ for some $q\le k=\log (n+1)$ be a valid sequence of groups. There is a polynomial time algorithm such that given $i\in [q]$, outputs either a root $r_i\in \rootset$ with the property that there are two sets  $A_{r_i},A'_{r_i}\subseteq V$ where $|A_{r_i}|\geq s_i$ and $|A'_{r_i}|\geq n - s_i +1$ where $\max\limits_{v\in A_{r_i}}\{c(v,r_i),c(r_i,v)\}\leq 2\cdot\maxlatrounded(\group_{i+1})$, and $\max\limits_{v\in A'_{r_i}}c(v,r_i)\leq 2\cdot\maxlatrounded(\group_{i+1})$, or outputs $\bot$ if no such root exists.
	\end{lemma}
	\begin{proof}
		Note that we can efficiently check for each $r\in \rootset$ whether it satisfies the properties mentioned in the lemma or not. If there is such an $r$, we output an arbitrary such root $r_i$; otherwise, we output $\bot$.
	\end{proof}
	
	In the following lemma, we prove if $\group^*_i$ is a tour-group then $r^*_i$ output by the previous lemma is close to vertices in $V(\group^*_i)$.
	
	\begin{lemma}\label{lemma:roots}
		Let $\group^*_i$ be a tour-group and $r^*_i$ be the root output by \cref{lem: computing a root}, then $r^*_i\neq \bot$ and we have
		\begin{equation}\label{eq:close_root}
			\max\{c(r^*_i,v),c(v,r^*_i)\} \le 4\maxlatrounded(\group^*_{i+1}) \quad \forall v \in V(\group^*_i) \enspace.
		\end{equation}
	\end{lemma}
	
	\begin{proof}
		Consider an arbitrary vertex $v^*\in V(\group^*_i)$. Since $\group^*_i$ is a tour-group, then $A_{v^*}:=V(\group^*_i)$ and $A'_{v^*}:=\bigcup_{j=1}^{i}V(\group^*_j)$ satisfy the properties in \cref{lem: computing a root}. 
		Therefore, \cref{lem: computing a root} returns a root $r^*_i\neq \bot$. Let $A_{r^*_i},A'_{r^*_i}$ be the sets guaranteed by \cref{lem: computing a root}.
		
		We conclude the proof by showing $\max\{c(r^*_i,v^*), c(v^*,r^*_i)\}\leq 4\maxlatrounded(\group^*_{i+1})$. Because of the sizes of the sets above, we must have $u_1\in A_{r^*_i}\cap A'_{v^*}\neq \emptyset$ and $u_2\in A_{v^*}\cap A'_{r^*_i}\neq\emptyset$. Therefore, we can write
		\[
		c(r^*_i,v^*) \le c(r^*_i,u_1)+c(u_1,v^*) \le 4\maxlatrounded(\group^*_{i+1})
		\]
		and
		\[
		c(v^*,r^*_i) \le c(v^*,u_2)+c(u_2,r^*_i) \le 4\maxlatrounded(\group^*_{i+1}),
		\]
		as desired.
	\end{proof}
	
	Now we are ready to define the time-intervals $I_1,\dots,I_q$ that we need to set up the \ref{eq:refined-LP}.
	
	\begin{definition}\label{def:buckets}
		Given a valid sequence of groups $\group_1,\ldots,\group_q$, define time-intervals $I_1:=[t_0,t_1),I_2:=[t_1,t_2),\dots,I_q:=[t_{q-1},t_q)$, where $t_0:= 0$, and 
		for $i \in [q]$, we set $t_i := t_{i-1} + 19\cdot \maxlatrounded(\group_{i+1})$. 
	\end{definition}
	
	
	We call $t_0,\ldots,t_q$ \emph{bucket thresholds}. We update the bound on $T$ as follows: $T=t_q + n^2  \cdot \max_{u,v \in V \cup \{s,\target\}} c(u,v)$.
	One property of these thresholds is that they form a geometric sequence. This will be crucial in rounding a fractional solution of our LP relaxation. Below, we prove this property. 
	
	\begin{lemma}
		Let $\group_1,\ldots,\group_{q}$ be a valid sequence of groups, and let $I_i=[t_{i-1},t_i)$ be the time-intervals obtained in \cref{def:buckets}. For $i \in [q-1]$, we have $t_{i+1} \ge \frac43 t_i$.
	\end{lemma}
	\begin{proof}
		Note that for $i \in [q-1]$, we have $t_{i+1} - t_{i} \ge t_i - t_{i-1}$.
		Moreover, for $i \in [q-2]$, we have $t_{i+2} - t_{i+1} \ge 2 (t_i - t_{i-1})$ by our grouping.
		Thus, for $i \in [q-1]$,
		\begin{align*}
			t_{i} &= t_i - t_{i-1} + \sum_{j=1}^{i-1} (t_j - t_{j-1}) \\
			&\le t_{i+1} - t_{i} + \frac12 \sum_{j=1}^{i+1} (t_j - t_{j-1}) \\
			&= \frac32 t_{i+1} - t_{i} \enspace.
		\end{align*}
	\end{proof}
	
	Now we are ready to present our enumeration algorithm for setting up the \ref{eq:refined-LP}.
	More precisely, we compute a set of triples $(\mathcal{T}=\{t_0,\ldots,t_q\},A_\textnormal{tour},R)$ so that for the triple that is compatible with $\OPT$, the \ref{eq:refined-LP} admits a solution of cost $O(1) \cdot \opt$, as we will show in the next subsection.
	
	\begin{algorithm}[H]
		\caption{Setting up the \ref{eq:refined-LP}}\label{alg: dirlat algorithm}
		{\bf Input:} A nice instance of \textsc{Directed Latency with Target} \\
		{\bf Output:} A set of triples $(\mathcal{T},A_\textnormal{tour},R)$.
		\begin{algorithmic}[1]
			\FOR{all valid sequences of groups $\group_1,\ldots,\group_q$ according to \cref{lem: number of valid groups}}
			\STATE Compute time-intervals $I_1,\ldots,I_q$ according to \cref{def:buckets}. 
			\FOR{each possible subset $A_{\textnormal{tour}}\subseteq [q]$}
			\STATE For each $i\in A_{\textnormal{tour}}$ compute $r_i$ according to \cref{lem: computing a root}.
			\IF{$r_i\neq\bot$ for each $i\in A_{\textnormal{tour}}$}
			\STATE Add $((t_0,\ldots,t_q),A_{\textnormal{tour}},(r_i)_{i \in A_{\textnormal{tour}}})$ to the output. 
			\ENDIF
			\ENDFOR
			\ENDFOR
		\end{algorithmic}
	\end{algorithm}
	
	The number of enumerations in lines 1 and 3 in the algorithm are polynomial in $n$ and $\frac{1}{\eps}$ by \cref{lem: number of valid groups} and the fact that $q\leq \log (n+1)$, respectively. Hence, the size of the output and the running time of \cref{alg: dirlat algorithm} is polynomial in $n$ and $\frac{1}{\eps}$.  
	
	\subsection{Existence of a cheap LP solution}\label{sub: cheap lp sol}
	Finally, we prove \cref{thm:existence}. Recall the definition of $\group^*_1,\ldots,\group^*_{q^*}$ based on $\OPT$, and let $I^*_1,\ldots,I^*_{q^*}$ be the time-intervals obtained from these groups according to \cref{def:buckets}. We show there is a subset $A^*_{\textnormal{tour}}\subseteq [q^*]$ such that the \ref{eq:refined-LP} defined based on $I^*_1,\ldots,I^*_{q^*}$, and roots $r^*_i$ for each $i\in A^*_{\textnormal{tour}}$ computed by \cref{lem: computing a root} has a feasible solution of cost at most $O(1)\cdot\opt$.
	
	\begin{proof}[Proof of \cref{thm:existence}]
		We define $A^*_{\textnormal{tour}}$ as follows: $i\in A^*_{\textnormal{tour}}$ if and only if there exist a vertex $v^*\in V(\group^*_{i+1})$ and some $u^*\in \bigcup_{j=1}^{i-1}V(\group^*_j)$ such that $c(v^*,u^*)\leq c(u^*,v^*)$. We say $v^*$ {\em certifies} $i\in A^*_{\textnormal{tour}}$. Note such an edge implies there is a tour of cost at most $2\cdot\maxlatrounded(\group^*_{i+1})$ on $V(\group^*_i)$, in other words, $\group^*_i$ is a tour-group. Hence, by \cref{lemma:roots}, $r^*_i$ computed in \cref{lem: computing a root} satisfies \eqref{eq:close_root}.
		Let $\roots^*$ denote the resulting set of roots.
		
		Recall $P^*_i$ is the subpath of $\OPT$ on the vertex set $V(\group^*_i)$. Define $P^*[u^*,v^*]$ to be the subpath of $\OPT$ from $u^*$ to $v^*$.
		We call the time-interval $I^*_i$ {\em tour-interval} if $i\in A^*_{\textnormal{tour}}$.
		
		We will now construct a walk $P$ on $V \cup \{s,\target\} \cup \roots$ that starts in $s$ and ends in $\target$ and is based on $\OPT$.
		We construct $P$ step by step following the time-intervals $I^*_i$.
		Let $i \in [q^*]$. 
		Assume that we already constructed $P$ for the first $i-1$ time-intervals. 
		\begin{itemize}
			\item If $I^*_i$ is not a tour-interval, we extend $P$ by the subpath of $P^*_i$ that is not yet in $P$. Note this subpath could be empty in which case we go to the next time-interval.
			\item If $I^*_i$ is a tour-interval, let $u$ denote the current endpoint of $P$.
			Moreover, let $u^*_i\in V(\group^*_i)$ be the first vertex on $P^*_i$ not yet in $P$, and let $w^*_i\in V(\group^*_{i+1})$ be the vertex with largest latency in $\OPT$ that certifies $i\in A^*_{\textnormal{tour}}$. We extend $P$ by adding $(u,r^*_i)$, $(r^*_i,u^*_i)$, $P^*[u^*_i,w^*_i]$, $(w^*_i,r^*_i)$ in this order. See \cref{fig: proof of good path} for an illustration of this construction.
		\end{itemize} 
		In both cases, let $P_i$ denote this extension of $P$.
		
		\begin{figure}
			\centering
			\begin{subfigure}{0.5\textwidth}
				\centering
				\includegraphics[scale=0.9]{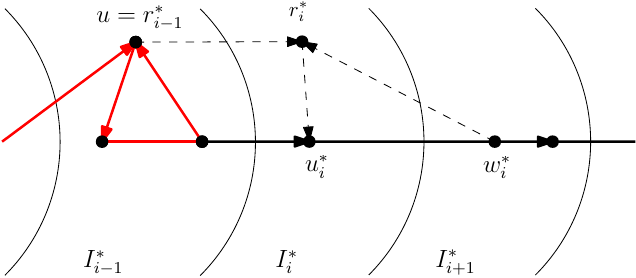}
				\caption{}
				\label{fig: path construction}
			\end{subfigure}
			\begin{subfigure}{0.5\textwidth}
				\centering
				\includegraphics[scale=0.9]{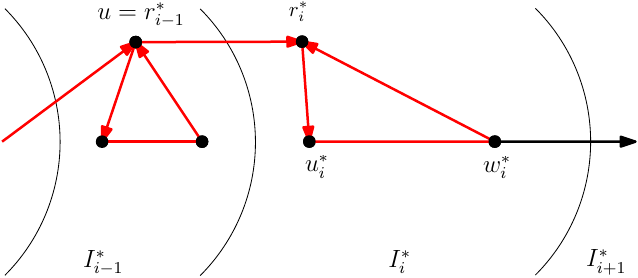}
				\caption{}
				\label{fig: final tour}
			\end{subfigure}
			\caption{(a) shows the construction of walk $P$ in the proof of \cref{thm:existence} up to $(i-1)$'th time interval with red color. Here we have both $i-1$ and $i$ are in $A^*_{\textnormal{tour}}$. The part of $\OPT$ not used so far is in black. The dashed edges are the added edges as in the proof. The edges of $\OPT$ are shown with a directed line segment, meanwhile the undirected line segments are a subpath of $\OPT$ (e.g., the line segment between $u^*_i$ and $w^*_i$). (b) shows the final walk $P$ constructed up to $i$'th time interval. Note that $I^*_i$ is changed, i.e., the threshold $t^*_i$ increased so that the embedding of $P$ in $G_T$ satisfies the constraints of the \ref{eq:refined-LP}.}
			\label{fig: proof of good path}
		\end{figure}
		
		Finally, extend $P$ by adding an edge to $\target$.
		Clearly, $P$ visits all vertices in $V$.
		We will show that we can embed $P$ in $G_T$. 
		More precisely, there is a walk $P_T$ in $G_T$ so that 
		\begin{itemize}
			\item for each $j \in [\ell]$, we have $u=v$ for the $j$'th vertex $u$ in $P$ and the $j$'th vertex $(v,t)$ in $P_T$, and
			\item the sub-walk $P_i$ of $P$ is embedded into $I^*_i$.
		\end{itemize}
		Such an embedding $P_T$ indeed exists by the following claim.
		
		\begin{claim}\label{claim:path_length}
			For each $i \in [q^*]$, we have $c(P_i) < t^*_i - t^*_{i-1}=19\cdot \maxlatrounded(\group_{i+1})$.
		\end{claim}
		\begin{proof}
			We first consider the case that $i \ge 2$ and $I^*_{i-1}$ and $I^*_i$ are tour-intervals.
			Note that in this case, the endpoint $u$ of $P_{i-1}$ equals $r^*_{i-1}$.
			Let $v$ be an arbitrary vertex in $V(\group^*_{i-1})$. Then,
			\begin{align*}
				c(u,r^*_i)&\leq c(r^*_{i-1},v) + c(v,u^*_i) + c(u^*_i,r^*_i) \\
				&\leq 4\maxlatrounded(\group^*_{i}) + c(v,u^*_i) + 4\maxlatrounded(\group^*_{i+1})\\
				&\leq 8\cdot \maxlatrounded(\group^*_{i+1}) + c(v,u^*_i),
			\end{align*}
			where the second inequality follows from applying \cref{lemma:roots} once for $i-1$ and once for $i$.
			
			Let $v'$ be a vertex in $\bigcup_{j=1}^{i-1}V(\group^*_j)$ such that $c(w^*_i,v') \le c(v',w^*_i)\leq \maxlatrounded(\group_{i+1})$. Note such a vertex exists because $I^*_i$ is a tour-interval. Then,
			\begin{align*}
				c(w^*_i,r^*_i) &\leq c(w^*_i,v') + c(v',u^*_i) + c(u^*_i,r^*_i) \\
				&\leq \maxlatrounded(\group^*_{i+1}) + \maxlatrounded(\group^*_{i+1}) + 4\cdot \maxlatrounded(\group^*_{i+1})\\
				&= 6\cdot \maxlatrounded(\group^*_{i+1}),
			\end{align*}
			where we again used \cref{lemma:roots} in the second inequality.
			
			Putting the above bounds together we get
			\begin{align*}
				c(P_i) &\le c(u,r^*_i) + c(r^*_i, u^*_i) +  c(P^*[u^*_i,w^*_i]) + c(w^*_i,r^*_i) \\
				&\le 8\cdot \maxlatrounded(\group^*_{i+1}) + c(v,u^*_i) + c(r^*_i,u^*_i) + c(P^*[u^*_i,w^*_i]) + 6\cdot \maxlatrounded(\group^*_{i+1})\\
				&\le 18\cdot\maxlatrounded(\group^*_{i+1}) + c(v,u^*_i) + c(P^*[u^*_i,w^*_i]) \\
				&< 19 \maxlatrounded(\group^*_{i+1}) \enspace,
			\end{align*} 
			where we used \cref{lemma:roots} in the second inequality to bound $c(r^*_i,u^*_i)\leq 4\cdot \maxlatrounded(\group^*_{i+1})$.
			Moreover, the last inequality holds because 
			\begin{equation*}
				c(v,u^*_{i}) + c(P^*[u^*_i,w^*_i]) \le \sum_{j=1}^{i+1} c(P^*_j) \le \maxlatrounded(\group^*_{i+1}) - c(s,v^*_1) < \maxlatrounded(\group^*_{i+1}) \enspace,
			\end{equation*}
			using the definition of nice instances in the final inequality.
			This finished the proof of \cref{claim:path_length}.
			Note that the bound improves in the case $i=1$ because then $u=s$. 
			Moreover, if $I^*_{i-1}$ or $I^*_i$ is not a tour-interval, the bound on $c(P_i)$ improves since we do not add the detour to the roots of $I^*_{i-1}$ or $I^*_i$.
		\end{proof}
		
		Next, we bound the objective value of the LP solution $(x,z)$ induced by $P_T$, where $z$ equals the characteristic vector of $E[P_T]$. For $i \in [q^*]$, we use the following useful bound on $t^*_i$: Let $A_i:=\{j\in[i]:~\maxlatrounded(\group^*_j)=\maxlatrounded(\group^*_{j+1})\}$. Then, we can write
		\begin{equation}\label{eq: bound on t_i}
			t^*_i=19\cdot \sum\limits_{j=1}^{i+1}\maxlatrounded(\group^*_j)\leq 38\cdot \sum\limits_{j\in[i+1]\setminus A_i}\maxlatrounded(\group^*_j)\leq 76 \cdot\maxlatrounded(\group^*_{i+1}), 
		\end{equation}
		where the last inequality follows from the fact that $(\maxlatrounded(\group^*_j))_{j\in[i+1]\setminus A_i}$ forms a geometric series with gap at least $2$. Since in $P_T$, each vertex $v \in V(\group^*_i)$ with $i \in [q^*]$ is visited up to time $t^*_i$, we can bound the objective value of $(x,z)$ by
		\begin{align*}
			\sum_{i=1}^{q^*} s(\group^*_i) \cdot t^*_i &\leq \sum_{i=1}^{q^*} \left(\sum_{j \in \group^*_i} \frac{n+1}{2^j} \right) \cdot 76\cdot\maxlatrounded(\group^*_{i+1})  \\
			&= \sum_{j=1}^k \frac{n+1}{2^j} \cdot 76\cdot\maxlatrounded(I^*_{j+1}) & \textnormal{(using $I^*_{k+1} := I^*_k$)} \\
			&= 304\cdot \sum_{j=3}^{k+2} \frac{n+1}{2^j} \cdot \maxlatrounded(I^*_{j-1}) \\
			&\le 304\cdot \left(2\opt + \frac{3}{4} \maxlatrounded(I^*_{k}) \right) \\
			&\le 836 \cdot \opt \enspace,
		\end{align*}
		where the first inequality follows from \eqref{eq: bound on t_i}, and the first equality holds by the definition of $\group^*_1,\dots,\group^*_{q^*}$.
		More precisely, for $j\in \group^*_i$, if $\maxlatrounded(I^*_{j})=\maxlatrounded(I^*_{j+1})$ then $\maxlatrounded(\group^*_{i+1})=\maxlatrounded(I^*_{j+1})$, and if $\maxlatrounded(I^*_{j})\neq\maxlatrounded(I^*_{j+1})$ then $j+1\in\group^*_{i+1}$ so also in this case $\maxlatrounded(\group^*_{i+1})=\maxlatrounded(I^*_{j+1})$. The second inequality follows from the fact that there are exactly $\frac{n+1}{2^j}$ vertices in $B^*_j$ with latency of at least  $\maxlat(I^*_{j-1}) \ge \frac12 \cdot \maxlatrounded(I^*_{j-1})$ each.
		
		Finally, we argue that $(x,z)$ is indeed feasible for the \ref{eq:refined-LP}. By our construction, $P_T$ starts in $s$, ends in $\target$, and visits each vertex in $V$.
		Thus, we have $(x,z) \in Q(G_T)$. 
		Moreover, for every $I^*_i$ where $i\in A^*_{\textnormal{tour}}$, we have that the first and the last vertex visited in the time interval $I^*_i$ is $r^*_i$. So it remains to show that all constraints of the last type in the \ref{eq:refined-LP} are satisfied. Consider a pair of vertices $u,v\in V$ and assume $(u,v)$ is not $t^*_{i+1}$-short for some $i\in [q^*-1]$. Then, we need to show
		\begin{equation*}
			\sum\limits_{t\in[t^*_{i}-1]}(x_{u,t}-x_{v,t})\leq \sum_{t \in I_\textnormal{tour}} (x_{u,t} +x_{v,t} ) \enspace.
		\end{equation*}
		
		If $\lat_{P_T}(u) \ge t^*_i$ or $\lat_{P_T}(v) < t^*_i$, then the LHS is at most $0$ so the inequality holds. So we can assume $\lat_{P_T}(u)<t^*_i\leq \lat_{P_T}(v)$. Furthermore, if either of $\lat_{P_T}(u)$ or $\lat_{P_T}(v)$ is in a tour-interval then the RHS is at least $1$ and the inequality holds as well.
		
		Note that by our construction of $P_T$, for each $i' \in [q^*]$, each vertex in $V(\group^*_{i'})$ is visited either in $I^*_{i'}$ or $I^*_{i'-1}$. Note that the latter case implies that $i' \ge 2$ and $I^*_{i'-1}$ is a tour-interval.
		
		Given the above discussion, we can assume $\lat_{P_T}(u)\in I^*_j$ and $\lat_{P_T}(v)\in I^*_h$ where $j\leq i < h$. Moreover, $u\in V(\group^*_{j})$ and $v\in V(\group^*_{h})$, as otherwise $u$ or $v$ would be visited in a tour-interval. 
		We cannot have $h\geq j+2$, as otherwise, using that $(u,v)$ is not $t^*_{i+1}$-short, $v$ would certify $h-1\in A^*_{\textnormal{tour}}$. Thus, $v$ would have been visited in $I^*_{h-1}$, a contradiction.
		
		Therefore, we have $j\leq i< h\leq j+1$ which implies $j=i$ and $h=i+1$. Hence, $c(u,v)\leq \maxlatrounded(\group^*_{i+1})<t^*_{i}$, a contradiction that $(u,v)$ is not $t^*_{i+1}$-short\footnote{We remark that the proof still works if we add constraints of the last type in the \ref{eq:refined-LP} for all $u,v \in V$ with $(u,v)$ not being $t_i$-short. However, to keep the discussion in \cref{sec:outline} more intuitive we did not opt for this version. Moreover, with our current analysis, this would not improve the guarantees we get when rounding fractional solutions as discussed in \cref{sec: rounding}.}.
	\end{proof}   
	
	\section{Computing paths on each bucket}\label{sec: preliminary}
	
	Before we prove \cref{thm:rounding}, we shortly recap how to use known LP-relative algorithms for \textsc{Asymmetric Path TSP} to construct a short path on each of our buckets as shown by \textcite{quasi-poly}.
	Although this section does not contain any novelties, we provide some brief discussion for the sake of completeness.
	
	Given a flow $z$ in the time-indexed graph $G_T$, we define its time-aggregation as follows.
	
	\begin{definition}
		Given an instance $(V \cup \{s,\target\},c)$ of \textsc{Directed Latency with Target}, let $G$ denote the complete directed graph on vertex set $V \cup \{s,\target\}$.
		For $z \in \mathbb{R}^{E[G_T]}$, define the \emph{time-aggregation} $f(z) \in \mathbb{R}^{E[G]}$ of $z$ by 
		\begin{equation*}
			f(z)_{(u,v)} := \sum_{e=((u,t),(v,t')) \in E[G_T]} z_e
		\end{equation*}
		for each $u,v \in V \cup \{s,\target\}$.
	\end{definition}
	
	Let $(x,z)$ denote a feasible solution to the \ref{eq:refined-LP}.
	To construct paths on each bucket, we consider the time-aggregation of parts of the flow $z$.
	By using a well-known tool from graph theory called \emph{splitting off}, these time-aggregations yield feasible solutions to some generalization of a natural LP relaxation of \textsc{Asymmetric Path TSP}.
	We briefly discuss this LP relaxation in \cref{sec:ATSP-LP} and the splitting off technique in \cref{sec:splitting-off}. 
	
	\subsection{The natural LP relaxation of  \textsc{Asymmetric Path TSP}}\label{sec:ATSP-LP}
	
	Given $\rho \in (0,1)$ and an asymmetric metric space $(V,c)$ with source $s \in V$ and target $t \in V$, we consider the following linear program, which for $\rho = 1$ is the natural linear programming relaxation for \textsc{Asymmetric Path TSP}.  
	
	\begin{equation}\tag{ATSPP$_\rho$-LP}\label{eq:atspp-lp}
		\begin{array}{rrcll}
			\min & \displaystyle \sum_{(u,v) \in V \times V} c_{u,v} x_{u,v} \\
			& x(\delta^+(v)) - x(\delta^-(v)) & = & \begin{cases}
				\phantom{-}1 & v = s \\
				-1 & v = t \\
				\phantom{-}0 & \text{else}
			\end{cases} & \forall v \in V \\
			& x(\delta(U)) & \ge & 2 \rho & \forall \emptyset \neq U \subseteq V \setminus \{s,t\} \\
			& x & \geq & 0 \enspace . & 
		\end{array}
	\end{equation}
	
	We use the following theorem by \textcite{quasi-poly}. 
	In their paper, they remark that they did not attempt to optimize the approximation factor, but favored simplicity of presentation.
	Moreover, they show that the integrality gap of \ref{eq:atspp-lp} is at least $\frac{1}{2\rho-1}$.
	Here and henceforth, $\atsp$ denotes the best-known LP-relative approximation factor for \textsc{Asymmetric TSP}. 
	By \cite{vygenATSP}, we have $\atsp < 15$.
	
	\begin{theorem}[Theorem 1.2 in \cite{quasi-poly}]\label{thm:ATSPP-rho}
		For $\rho > \frac12$, there is an LP-relative $\frac{\psi}{2\rho-1}$-approximation algorithm for \ref{eq:atspp-lp}, where  $\psi \le 1 + 32 \atsp$.
	\end{theorem}
	
	\subsection{Directed splitting off}\label{sec:splitting-off}
	
	In this subsection we discuss a powerful and well-known tool in graph theory called \emph{splitting off}. 
	First introduced for undirected graphs by \textcite{lovasz_undirected}, this technique was extended to directed graphs by \textcite{mader_digraph}.
	Let a digraph $G=(V,E)$ be given, and let $v \in V$. 
	Under certain assumptions, we can, for every $(u,v) \in \delta^-(v)$, efficiently find an edge $(v,w) \in \delta^+(v)$, so that replacing $(u,v)$ and $(v,w)$ by $(u,w)$ maintains certain connectivity properties.
	We use the following formulation from \cite{Traub_Vygen_2024}, which is based on work by \textcite{jackson} and \textcite{frank}.
	
	\begin{theorem}[Theorem 3.3 in \cite{Traub_Vygen_2024}]\label{thm:splitting_off}
		Let $G=(V,E)$ be a digraph with weights $x: E \rightarrow \mathbb{R}_{\ge 0}$ that satisfy $x(\delta^-(v)) = x(\delta^+(v))$ for all $v \in V$. Let $T \subsetneq V$ and $z \in V \setminus T$. 
		Let $\lambda > 0$ and 
		\begin{equation}\label{eq:splitting_off}
			x(\delta(U)) \ge \lambda \quad \forall U \subsetneq V : T \cap U \neq \emptyset, T \setminus U \neq \emptyset \enspace.
		\end{equation}
		Then one can compute in polynomial time a list of triples $(e_i,f_i,\gamma_i) \in \delta^-(z) \times \delta^+(z) \times \mathbb{R}_{>0}$, $i=1,\dots,k$ such that splitting off all $(e_i,f_i,\gamma_i)$ maintains \eqref{eq:splitting_off} and leads to $x(e) = 0$ for all edges $e$ incident to $z$. Here splitting off $(e,f,\gamma)$ means reducing $x(e)$ and $x(f)$ by $\gamma$ and adding an edge with weight $\gamma$ from the tail of $e$ to the head of $f$. 
	\end{theorem}
	
	We will apply splitting off to solutions of the \ref{eq:atspp-lp} relaxation. 
	Thus, we use the following variant, which is a direct consequence of \cref{thm:splitting_off}.
	
	\begin{theorem}\label{thm:splitting_off_path}
		Let $G=(V,E)$ be a digraph and let $s,t \in V$. 
		Let $x: E \rightarrow \mathbb{R}_{\ge 0}$ satisfying
		\begin{align*}
			x(\delta^+(v)) - x(\delta^-(v))  =  \begin{cases}
				\phantom{-}1 & v = s \\
				-1 & v = t \\
				\phantom{-}0 & \text{else} 
			\end{cases}
		\end{align*}
		for all $v \in V$.
		Let $T \subseteq V \setminus \{s,t\}$  and $z \in V \setminus (\{s,t\} \cup T)$. 
		Let $\lambda \in (0,2]$ and 
		\begin{equation}\label{eq:splitting_off_path}
			x(\delta(U)) \ge \lambda \quad \forall U \subseteq V \setminus \{s,t\} : T \cap U \neq \emptyset \enspace.
		\end{equation}
		Then one can compute in polynomial time a list of triples $(e_i,f_i,\gamma_i) \in \delta^-(z) \times \delta^+(z) \times \mathbb{R}_{>0}$, $i=1,\dots,k$ such that splitting off all $(e_i,f_i,\gamma_i)$ maintains \eqref{eq:splitting_off_path} and leads to $x(e) = 0$ for all edges $e$ incident to $z$. Here splitting off $(e,f,\gamma)$ means reducing $x(e)$ and $x(f)$ by $\gamma$ and adding an edge with weight $\gamma$ from the tail of $e$ to the head of $f$. 
	\end{theorem}
	
	\begin{proof}
		Increase $x((t,s))$ by $1$ and let $x'$ denote the resulting edge weights.
		Note that $x'$ satisfies \eqref{eq:splitting_off} since $x$ sends one unit of flow from $s$ to $t$ and $\lambda \le 2$.
		Now, apply \cref{thm:splitting_off} to $x'$, $\lambda$, $T$, and $z$. 
		
		Note that splitting off the resulting list of triples $(e_i,f_i,\gamma_i) \in \delta^-(z) \times \delta^+(z) \times \mathbb{R}_{>0}$, $i=1,\dots,k$ leads to $x(e) = 0$ for each edge $e$ incident to $z$ since $x(e)=x'(e)$ for each such $e$.
		Moreover, \eqref{eq:splitting_off_path} is still satisfied since $x(\delta(U)) = x'(\delta(U))$ for each $U \subseteq V \setminus \{s,t\}$.
	\end{proof}
	
	\section{Rounding Algorithm}\label{sec: rounding}
	
	In this section we prove \cref{thm:rounding}.
	Let $(V \cup \{s,\target\},c)$ be an instance of \textsc{Directed Latency with Target} and let $(x,z)$ be a feasible solution to the \ref{eq:refined-LP}. We show how to round $(x,z)$ to get a feasible integral solution.
	
	Let $\delta \in (0,1)$ be some constant that we will choose later.
	Recall $A_{\textnormal{tour}}\subseteq [q]$ is a set of indices of tour-intervals, and $I_\textnormal{tour} := \bigcup_{i \in A_\textnormal{tour}} I_i$.
	Let $V_{\textnormal{tour}}$ be the set of vertices $v \in V$ with $\sum_{t \in I_{\textnormal{tour}}} x_{v,t} \ge \delta$.
	For ease of notation, we write $\ub(t) := t_i$ for $i \in [q]$ and $t \in I_i$. With this notation, we can write the LP objective as $\min \sum_{v \in V, t \in [T]} \ub(t) x_{v,t}$.
	
	We start by constructing tours rooted in $(r_i)_{i \in A_{\textnormal{tour}}}$ covering all vertices in $V_{\textnormal{tour}}$. 
	The proof of the following lemma is based on the observation that, thanks to our new LP constraints, the time-aggregation of $z$ restricted to a tour-interval is a circulation.
	Thus, by simply scaling $z$ on tour-intervals, we can boost the extent to which vertices are visited within tour-intervals.
	Recall that we refer to cycles as tours.
	
	\begin{lemma}\label{lemma:rounding-tour-buckets}
		Let $\rho_1 \in (0,1)$. 
		We can efficiently compute tours $(T_i)_{i \in A_{\textnormal{tour}}}$ so that $V_{\textnormal{tour}} \subseteq \bigcup_{i \in A_{\textnormal{tour}}} V[T_i]$ and for each $i \in A_{\textnormal{tour}}$
		\begin{enumerate}
			\item\label{item:root} $r_i \in V[T_i]$, 
			\item\label{item:latency_tour_buckets} $\sum_{t \in [T]} \ub(t) x_{v,t} \ge (1-\rho_1) \delta \cdot t_i$ for each $v \in V[T_i] \cap V$, and
			\item\label{item:cost_tour_buckets} $c(T_i) \le \left(1 + \frac{\atspp}{\rho_1\delta} \right) \cdot t_i$,
		\end{enumerate}
		where $\atspp$ denotes the approximation guarantee of the best-known LP-relative approximation algorithm for \textsc{Asymmetric Path TSP}. 
	\end{lemma}
	
	\begin{proof}
		For each $i \in A_{\textnormal{tour}}$, let $B_i$ be the set of vertices $v \in V_{\textnormal{tour}}$ for which $i$  is the smallest index in $A_{\textnormal{tour}}$ with 
		\begin{equation*}
			\sum_{j \in A_{\textnormal{tour}}: j \le i} \sum_{t \in I_j} x_{v,t} \ge \rho_1 \delta \enspace .   
		\end{equation*}
		Then, for $v \in B_i$ we have
		\begin{equation*}
			\sum_{t \in [T]} \ub(t) x_{v,t} \ge \sum_{j \in A_{\textnormal{tour}}: j \geq i} \sum_{t \in I_j} t_i x_{v,t} \ge (1-\rho_1) \delta t_i \enspace.
		\end{equation*}
		
		It remains to show that we can efficiently construct a tour on $B_i \cup \{r_i\}$ of length at most $\left(1 + \frac{\atspp}{\rho_1\delta} \right) \cdot t_i$.
		
		For time interval $I$, let $z_I$ denote the restriction of $z$ to $I$, i.e., restrict $z$ to edges $\big((u,t),(v,t')\big) \in E[G_T]$ such that $t,t'\in I$.
		
		Now, for $i \in A_{\textnormal{tour}}$ consider the sum of the time-aggregated flows
		\begin{equation*}
			f_i := f(z_{[0,t_i)}) + \left( \frac{1}{\rho_1 \delta} -1 \right) \sum_{j \in A_{\textnormal{tour}}: j \le i} f(z_{I_j}) \enspace.
		\end{equation*}
		Clearly, $c(f_i) \le \frac{t_i}{\rho_1\delta}$. Moreover, because of the constraint $z_e=0 \enspace \forall e\in F_{\roots}$ we have that $f(z_{I_j})$ is a circulation for all $j\in A_{\textnormal{tour}}$. Hence $f_i$ is an $s$-$r_i$-flow of value $1$.
		Note that via splitting off $f_i$, we get a feasible solution to the \ref{eq:atspp-lp} for $\rho =1$ on $V_i := B_i \cup \{r_i,s\}$ with starting point $s$ and endpoint $r_i$ of cost at most $\frac{t_i}{\rho_1\delta}$.
		Indeed, we can iteratively apply \cref{thm:splitting_off_path} with $\lambda = 2$, $T = B_i$ and $z \in (V \cup \roots) \setminus V_i$ to split off all vertices in $(V \cup \roots) \setminus V_i$.
		Hence, we can compute an $s$-$r_i$-path $P_i$ visiting all vertices in $B_i$ of length at most $\frac{\atspp t_i}{\rho_1\delta}$.
		
		Note that for each $u \in B_i$ there is some $t \in I_i$ with $x_{u,t}>0$. 
		Thus, by the LP constraints, $z$ sends some flow from $(r_i,t')$ to $(u,t)$ with $t' \le t < t_i$, yielding $c(r_i, u) < t_i$.
		Hence, we can turn $P_i$ into a tour on $B_i \cup \{r_i\}$ by adding the edge $(r_i, u)$ and deleting the edge $(s,u)$, where $u$ is the vertex directly visited after $s$ on $P_i$. 
		This increases the length of $P_i$ by at most $t_i$, finishing the proof.
	\end{proof}
	
	Using \cref{lemma:rounding-tour-buckets}, we can construct a path starting in $s$ and visiting all vertices in $V_{\textnormal{tour}}$ so that its total latency is not much larger than what the LP solution pays for vertices in $V_{\textnormal{tour}}$.
	It remains to show how to include the vertices in $V \setminus V_{\textnormal{tour}}$.
	
	For a set of paths $(P_i)_{i \in [q] \setminus A_{\textnormal{tour}}}$ (possibly $P_i$ does not exist for some indices $i$ in which case we write $V[P_i]=\emptyset$), and $i \in [q] \setminus A_{\textnormal{tour}}$ with $V[P_i] \neq \emptyset$, we denote by $\textnormal{succ}(i)$ the minimal $j > i$ such that $j \in A_{\textnormal{tour}}$ or $V[P_j] \neq \emptyset$. 
	We set $\textnormal{succ}(i) = q+1$ if there is no such $j$.
	Analogously, we denote by $\textnormal{pred}(i)$ the maximal index $j < i$ such that $j \in A_{\textnormal{tour}}$ or $V[P_j] \neq \emptyset$. 
	We set $\textnormal{pred}(i) = 0$ if there is no such $j$.\footnote{The notion of $\textnormal{succ}(i)$ and $\textnormal{pred}(i)$ are needed because in the proof of the next lemma we might end up with some empty buckets so we want to skip these buckets. Hence, on a first reading, we encourage the reader to assume that all buckets are non-empty. In this case, $\textnormal{succ}(i)$ becomes $i+1$ and similarly $\textnormal{pred}(i)$ becomes $i-1$.}
	
	\begin{lemma}\label{lemma:rounding-nontour-buckets}
		Let $\rho_2 \in (0,1)$ with $\rho_2(1-\delta) - 3\delta > \frac12$. 
		We can efficiently compute paths $(P_i)_{i \in [q] \setminus A_{\textnormal{tour}}}$ so that $\bigcup_{i \in [q] \setminus A_{\textnormal{tour}}} V[P_i] = V \setminus V_{\textnormal{tour}}$ and for each $i \in [q] \setminus A_{\textnormal{tour}}$ with $V[P_i] \neq \emptyset$,
		\begin{enumerate}
			\item\label{item:short_edge} if $\textnormal{succ}(i) \in [q] \setminus A_{\textnormal{tour}}$, then $c(v_i,u_{\textnormal{succ}(i)}) \le 2 t_{\textnormal{succ}(i)}$ for the endpoint $v_i$ of $P_i$ and the starting point $u_{\textnormal{succ}(i)}$ of $P_{\textnormal{succ}(i)}$,
			\item\label{item:frac_visited_endpoint} $\sum_{t \in [t_i-1]} x_{v_i,t} >0$ for the endpoint $v_i$ of $P_i$,
			\item\label{item:frac_visited_startpoint} if $i=1$ or $\textnormal{pred}(i) \in A_{\textnormal{tour}}$, then $c(r_{\textnormal{pred}(i)},u_i) \le t_i$ for the starting point $u_i$ of $P_i$ (using $r_0 = s$),
			\item\label{item:latency} $\sum_{t \in [T]} \ub(t) x_{v,t} \ge (1-\rho_2) (1-\delta) \cdot t_i$ for each $v \in V[P_i]$, and
			\item\label{item:cost} $c(P_i) \le \frac{\psi}{2\rho_2(1-\delta) - 6\delta -1} \cdot t_i$,
		\end{enumerate}
		where $\psi \le 1 + 32 \atsp$ is as in \cref{thm:ATSPP-rho}. 
	\end{lemma}
	We remark that the first three properties are used later to bound the extra cost resulting from concatenating the tours and paths given by \cref{lemma:rounding-tour-buckets} and \cref{lemma:rounding-nontour-buckets}, respectively.
	To guarantee these properties, we use some refined construction of buckets, leveraging the last class of constraints in the \ref{eq:refined-LP}.
	
	\begin{proof}
		We start with constructing the vertex sets on which we then construct the paths.
		
		For $i \in [q] \setminus A_{\textnormal{tour}}$, we define $B_i$ to be the vertices  $v \in V \setminus V_{\textnormal{tour}}$ for which $i$  is the smallest index in $[q] \setminus A_{\textnormal{tour}}$ with 
		\begin{equation*}
			\sum_{j \in [q] \setminus A_{\textnormal{tour}}: j \le i} \sum_{t \in I_j} x_{v,t} \ge \rho_2 (1-\delta) \enspace .   
		\end{equation*}
		Note that $\bigcup_{i \in [q] \setminus A_\textnormal{tour}} B_i = V \setminus V_\textnormal{tour}$ by definition of $V_\textnormal{tour}$. The sets $B_i$ are preliminary buckets. In the following, we describe how we recursively define the final buckets $\overline{B}_i$.
		
		For $i \in [q] \setminus A_{\textnormal{tour}}$,
		we denote by $\textnormal{pred}_{\overline{B}}(i)$ the maximal index $j < i$ such that $j \in A_{\textnormal{tour}}$ or $\overline{B}_j \neq \emptyset$.
		We set $\textnormal{pred}_{\overline{B}}(i) = 0$ if there is no such $j$. 
		Note that, as we will exploit later, we have $\textnormal{pred}(i)=\textnormal{pred}_{\overline{B}}(i)$ for any paths on these final buckets.
		
		Let $W_i \subseteq V \setminus V_{\textnormal{tour}}$ be defined as follows:
		\begin{stepsalph}
			\item\label{item:W-a} If $i+1 \in A_{\textnormal{tour}}$ and $\textnormal{pred}_{\overline{B}}(i) \in A_{\textnormal{tour}}$, let $W_i$ be the set of vertices $w \in V \setminus V_{\textnormal{tour}}$ for which $\sum_{t \in I_i} x_{w,t} >0$ and
			\begin{equation*} 
				\sum_{t \in [t_i-1]} x_{w,t} \ge \rho_2 (1-\delta) - 3 \delta \enspace.
			\end{equation*}  
			\item\label{item:W-b} Else, if $i+1 \in A_{\textnormal{tour}}$, let $W_i$ be the set of vertices $w \in V \setminus V_{\textnormal{tour}}$ for which 
			\begin{equation*} 
				\sum_{t \in [t_i-1]} x_{w,t} \ge \rho_2 (1-\delta) - 3 \delta \enspace.
			\end{equation*}  
			\item\label{item:W-c} Otherwise, let $W_i$ be the set of vertices $w \in V \setminus V_{\textnormal{tour}}$ for which there is some $v \in B_i$ so that $(v,w)$ is not $t_{i+1}$-short.
		\end{stepsalph}
		We define 
		\begin{equation*}
			\overline{B}_i := (B_i \cup W_i) \setminus \bigcup_{j \in [q] \setminus A_{\textnormal{tour}}: j < i} \overline{B}_j \enspace.
		\end{equation*}
		
		We remark that some of these sets might be empty.
		We start with the following observation.
		
		\begin{claim}\label{claim:frac_visited}
			Let $i \in [q] \setminus A_{\textnormal{tour}}$ and $v \in \overline{B}_i$.
			Then, $\sum_{t \in [t_i-1]} x_{v,t} \ge \rho_2(1-\delta) - 3\delta$.
			If $i+1 \in [q] \setminus A_{\textnormal{tour}}$, then even $\sum_{t \in [t_i-1]} x_{v,t} \ge \rho_2(1-\delta) - 2\delta$.
		\end{claim}
		\begin{proof}[Proof of Claim \ref{claim:frac_visited}]
			Clearly, this holds for each vertex in $B_i$ by definition. 
			If $i+1 \in A_{\textnormal{tour}}$, this holds for each vertex in $W_i$ by definition.
			Moreover, note that $W_q = \emptyset$.
			Thus, it remains to consider the case $v \in W_i$ and $i+1 \in [q] \setminus A_{\textnormal{tour}}$.
			Let $u \in B_i$ so that $(u,v)$ is not $t_{i+1}$-short.
			By the LP constraints, we have
			\begin{align*}
				\sum_{t \in [t_i-1]} x_{v,t}
				& \ge \sum_{t \in [t_i-1]} x_{u,t} - \sum_{t \in I_{\textnormal{tour}}} \left( x_{u,t} + x_{v,t} \right) \\
				& \ge \rho_2 (1-\delta) - 2 \delta \enspace.
			\end{align*}
			
		\end{proof}

		For each $i \in [q] \setminus A_{\textnormal{tour}}$, we will construct a path $P_i$ on $\overline{B}_i$ so that  the resulting set of paths has the desired properties.
		We start with showing \ref{item:short_edge}-\ref{item:latency}.
		
		\begin{claim}\label{claim: four properties}
			For $i \in [q] \setminus A_{\textnormal{tour}}$, let $P_i$ be a path on $\overline{B}_i$. 
			Then, these paths satisfy the properties \ref{item:short_edge}-\ref{item:latency}.
		\end{claim}
		
		\begin{proof}[Proof of Claim \ref{claim: four properties}]
			Let $i \in [q] \setminus A_{\textnormal{tour}}$ with $\overline{B}_i \neq \emptyset$.
			We start with showing \ref{item:short_edge}.
			So, let $\textnormal{succ}(i) \in [q] \setminus A_{\textnormal{tour}}$.
			Note that this implies $i+1 \in [q] \setminus A_{\textnormal{tour}}$.
			Thus, for $v \in B_i$, we observe that by definition of $W_i$, the edge $(v,u_{\textnormal{succ}(i)})$ is $t_{i+1}$-short.
			Moreover, by \cref{claim:frac_visited}, $\sum_{t \in [t_{\textnormal{succ}(i)}-1]} x_{u,t} \ge \rho_2(1-\delta) - 3\delta > \frac12$ for $u \in \{v, u_{\textnormal{succ}(i)}\}$.
			Hence, $z$ restricted to the time interval $[0,t_{\textnormal{succ}(i)})$ sends some flow from $v$ to $u_{\textnormal{succ}(i)}$ or vice versa.
			Thus, $\min \{c(v, u_{\textnormal{succ}(i)}), c(u_{\textnormal{succ}(i)},v)\} < t_{\textnormal{succ}(i)}$.
			Together, this yields $c(v, u_{\textnormal{succ}(i)}) \le t_{\textnormal{succ}(i)}$.
			In the following we will show that, if $v_i \in W_i$, then there is some $v' \in B_i$ with $c(v_i,v') \le t_{i}$, which then finishes the proof of \ref{item:short_edge}.
			
			So, assume $v_i \in W_i$.
			By definition of $W_i$ there is some $v' \in B_i$ so that the edge $(v',v_i)$ is not $t_{i+1}$-short.
			Thus, $c(v_i,v') < c(v',v_i)$.
			By \cref{claim:frac_visited}, we have $\sum_{t \in [t_{i}-1]} x_{u,t} \ge \rho_2(1-\delta) - 3\delta > \frac12$ for $u \in \{v_i, v'\}$.
			Hence, $z$ restricted to the time interval $[0,t_i)$ sends some flow from $v_i$ to $v'$ or vice versa, yielding $\min\{c(v_i,v'), c(v',v_i)\} \le t_{i}$.
			Together, we get $c(v_i,v') \le t_{i}$ as desired.
			
			Note that \ref{item:frac_visited_endpoint} directly follows by \cref{claim:frac_visited} and $\rho_2(1-\delta) - 3\delta > 0$.
			
			Next, we show \ref{item:frac_visited_startpoint}.
			If $i=1$, then, by \cref{claim:frac_visited} and $\rho_2(1-\delta) - 3\delta > 0$,
			$z$ sends some flow from $s$ to $u_i$ up to time $t_i$.
			So assume $\textnormal{pred}(i) = \textnormal{pred}_{\overline{B}}(i) \in A_\textnormal{tour}$.
			Note that it suffices to show $\sum_{t \in I_i} x_{u_i,t} > 0$ since then,  by the LP-constraints, $z$ sends some flow from $r_{\textnormal{pred}(i)}$ to $u_i$ up to time $t_i$.
			If $u_i \in B_i$, this directly follows by definition of $B_i$. 
			It remains to consider the case $u_i \in W_i$.
			If $i$ falls into category \ref{item:W-a}, we are also done by definition of $W_i$.
			It remains to consider the case that $i$ falls into category \ref{item:W-c}.
			Let $j \in [q]$ with $j < i$ be as large as possible so that either $j$ falls into category \ref{item:W-a} and $\sum_{t \in I_j} x_{u_i,t} >0$, or falls into category \ref{item:W-b}.
			If no such $j$ exists, with the fact that $\textnormal{pred}(i) \in A_{\textnormal{tour}}$, we get $\sum_{t \in I_j} x_{u_i,t} = 0$ for each $j \in [i-1] \setminus A_\textnormal{tour}$.
			Thus, 
			\begin{equation*}
				\sum_{t \in I_i} x_{u_i,t} = \sum_{t \in [t_i-1] \setminus I_\textnormal{tour}} x_{u_i,t} 
				\ge \sum_{t \in [t_i-1]} x_{u_i,t} - \delta 
				\ge \rho_2(1-\delta) - 3\delta 
				>0 \enspace,
			\end{equation*}
			where we used $u_i \in V \setminus V_\textnormal{tour}$ in the first inequality, and \cref{claim:frac_visited} in the second inequality.
			Otherwise, if such $j$ exists, we have $\sum_{t \in I_{\ell}} x_{u_i,t} = 0$ for each $\ell \in [i-1] \setminus A_\textnormal{tour}$ with $\ell > j$ by maximality of $j$.
			Thus, using $u_i \in W_i \setminus W_j$,
			\begin{align*}
				\sum_{t \in I_i} x_{u_i,t} &\ge \sum_{t \in [t_i-1]} x_{u_i,t} - \sum_{t \in [t_j-1]} x_{u_i,t} - \sum_{t \in I_\textnormal{tour}} x_{u_i,t} \\
				&\ge \left( \rho_2(1-\delta) - 2\delta \right) - \sum_{t \in [t_j-1]} x_{u_i,t} - \delta \\
				&> \left( \rho_2(1-\delta) - 2\delta \right) - \left( \rho_2(1-\delta) - 3\delta \right) -\delta \\
				&= 0 \enspace,
			\end{align*}
			where we used  \cref{claim:frac_visited} and $u_i \in V \setminus V_\textnormal{tour}$ in the second inequality, and $u_i \notin W_j$ in the final inequality.
			
			We finally show \ref{item:latency}. Observe that for each $i \in [q] \setminus A_{\textnormal{tour}}$ and each $v \in \overline{B}_i$, we have
			\begin{equation*}
				\sum_{t \in [T]} \ub(t) x_{v,t} \ge \sum_{j \in [q] \setminus A_{\textnormal{tour}}: j \ge i} \sum_{t \in I_j} t_i x_{v,t} \ge (1-\rho_2) (1-\delta) t_i \enspace,
			\end{equation*}
			where the last inequality follows because if $v\in \overline{B}_i$ then $v\in B_j$ for some $j\geq i$ together with the definition of $B_j$'s. This yields \ref{item:latency}.
		\end{proof}
		
		It remains to show that for each $i \in [q] \setminus A_{\textnormal{tour}}$, we can construct a path $P_i$ on $\overline{B}_i$ satisfying \ref{item:cost}.
		
		Consider the time-aggregated flow $f$ of $z$ restricted to the time interval $[0,t_i)$.
		Add an artificial vertex $v^*$ and add edges $e=(u,v^*)$ with $c(e) = 0$ for $u \in V \cup \{s\} \cup \roots$. 
		Set $f_e$ on all new edges so that $f$ is an $s$-$v^*$-flow of value $1$.
		More precisely, for $u \in V \cup \{s\} \cup R$, set 
		\begin{equation*}
			f_{(u,v^*)} := \sum_{e=((u,t'),(v,t)) \in E[G_T] : t' < t_i, t \ge t_i} z_e \enspace.
		\end{equation*}
		
		Clearly, $c(f) \le t_i$.
		Note that each vertex in $\overline{B}_i$ is visited by $f$ to an extent of at least $\rho_2(1-\delta) - 3\delta$ by \cref{claim:frac_visited}.
		Thus, by iteratively applying splitting off, we get a feasible solution $f'$ to the \ref{eq:atspp-lp} for $\rho = \rho_2(1-\delta) - 3\delta$ on $V_i := \overline{B}_i \cup \{s,v^*\}$ with starting point $s$ and endpoint $v^*$ of cost at most $t_i$.
		Indeed, we can iteratively apply \cref{thm:splitting_off_path} with $\lambda = 2\rho$, $T = \overline{B}_i$ and $z \in (V \cup \roots) \setminus V_i$ to split off all vertices in $(V \cup \roots) \setminus V_i$.
		Hence, using \cref{thm:ATSPP-rho}, we can compute an $s$-$v^*$-path visiting all vertices in $\overline{B}_i$ of length at most $\frac{\psi}{2\rho_2(1-\delta) - 6\delta -1} \cdot t_i$.
		Let $P_i$ denote the path on $\overline{B}_i$ resulting from deleting the first and the last edge of this path.
		Note that $P_i$ satisfies (v).
	\end{proof}
	
	Now, we are able to prove \cref{thm:rounding}.
	
	\begin{proof}[Proof of \cref{thm:rounding}]
		Use \cref{lemma:rounding-tour-buckets} and \cref{lemma:rounding-nontour-buckets} to compute a set of tours $(T_i)_{i \in A_{\textnormal{tour}}}$ and a set of paths $(P_i)_{i \in [q] \setminus A_{\textnormal{tour}}}$, respectively.
		Now, for $i \in A_{\textnormal{tour}}$, interpret $T_i$ as a closed walk $P_i$ starting and ending in $r_i$.
		We will construct the final path $P$ by starting in $s$, concatenating all $(P_i)_{i \in [q]}$, and then shortcutting all artificial roots.
		
		Clearly, $V[P] = V \cup \{s\}$. 
		To compute the total latency of $P$, we show that for each $i \in [q]$, the distance of the endpoint $v_{\textnormal{pred}(i)}$ of $P_{\textnormal{pred}(i)}$ (using $v_{0} := s$) to the starting point $u_i$ of $P_i$ is at most $2t_i$. 
		We consider the following case distinction.
		\begin{itemize}
			\item If $i = 1 \in A_{\textnormal{tour}}$, then $u_1 = r_1$.
			By the LP constraints, there must be some $s$-$r_1$-flow in $z$ restricted to $[0,t_1)$, and thus $c(s,r_1) \le t_1$.
			\item If $i \in [q] \setminus A_{\textnormal{tour}}$ and either $i=1$ or $\textnormal{pred}(i) \in  A_{\textnormal{tour}}$, then $c(v_{\textnormal{pred}(i)}, u_i) \le t_i$ by \cref{lemma:rounding-nontour-buckets}~\ref{item:frac_visited_startpoint}.
			\item If $\textnormal{pred}(i),i \in A_{\textnormal{tour}}$, then $v_{\textnormal{pred}(i)} = r_{\textnormal{pred}(i)}$ and $u_i = r_i$.
			By the LP constraints, there must be some $r_{\textnormal{pred}(i)}$-$r_i$-flow in $z$ restricted to $[0,t_i)$, and thus $c(r_{\textnormal{pred}(i)},r_i) \le t_i$.
			\item If $\textnormal{pred}(i),i \in [q] \setminus A_{\textnormal{tour}}$, then $c(v_{\textnormal{pred}(i)}, u_i) \le 2t_i$ by \cref{lemma:rounding-nontour-buckets}~\ref{item:short_edge}.
			\item If $\textnormal{pred}(i) \in [q] \setminus A_{\textnormal{tour}}$ and $i \in A_{\textnormal{tour}}$, note that $ \sum_{t \in [t_i-1]} x_{v_{\textnormal{pred}(i)},t} >0$, as guaranteed by \cref{lemma:rounding-nontour-buckets}~\ref{item:frac_visited_endpoint}. 
			Moreover, $u_i = r_i$. 
			Hence, by the LP constraints, there must be some $v_{\textnormal{pred}(i)}$-$r_i$-flow in $z$ restricted to $[0,t_i)$, and thus $c(v_{\textnormal{pred}(i)},r_i) \le t_i$.
		\end{itemize}
		Hence, for $v \in V[P_i]$, its latency in $P$ can be bounded by
		\begin{align*}
			c_P(v) &\le \sum_{j=1}^i (2t_j + c(P_j)) \\
			&\le \sum_{j=1}^i \max \left\{ 3 + \frac{\atspp}{\rho_1\delta} \enspace, 2 + \frac{\psi}{2\rho_2(1-\delta) - 6\delta -1} \right\} t_j \\
			&< \left(\sum_{j=0}^{\infty} \left(\frac34\right)^{j} \right) \max \left\{ 3 + \frac{\atspp}{\rho_1\delta} \enspace, 2 + \frac{\psi}{2\rho_2(1-\delta) - 6\delta -1} \right\} t_i \\
			&\le 4 \max \left\{ 3 + \frac{\atspp}{\rho_1\delta} \enspace, 2 + \frac{\psi}{2\rho_2(1-\delta) - 6\delta -1} \right\} \\
			&\phantom{=} \cdot \max \left\{ \frac{1}{\delta(1-\rho_1)} \enspace, \frac{1}{(1-\delta)(1-\rho_2)} \right\} \sum_{t \in [T]} \ub(t) x_{v,t} \enspace,
		\end{align*}
		where, in the second and fourth inequality, we used the guarantees given by \cref{lemma:rounding-tour-buckets} and \cref{lemma:rounding-nontour-buckets}, and in the third inequality, we used that $t_{i+1} \ge \frac43 t_{i}$ for each $i \in [q-1]$.
		Using $\atsp, \atspp < 15$  \cite{vygenATSP}, $\psi \le 1 + 32 \atsp$, and choosing $\delta = 0.063$, $\rho_1 = 0.196$, and $\rho_2 = 0.946$, yields the desired bound of $96473$.
	\end{proof}
	
	
	\begingroup
	\setlength{\emergencystretch}{1em}
	\printbibliography
	\endgroup
	
	\appendix
	
	\section{Proof of \cref{thm: reduction to bounded instance}}\label{Appendxi A}
	
	In the proof of Theorem 1 in Appendix of \cite{quasi-poly}, Friggstad and Swamy present a reduction from general instances of \textsc{Directed Latency} to instances where each $c(u,v)$ is a positive integer polynomially bounded in $n$ and $\frac{1}{\eps}$ by losing a factor of $(1+\eps)$ in the approximation. We remark that in their proof, they increase the cost of small edges (even the $0$ cost edges) to positive integer.
	We closely follow their proof, exploiting that their reduction does not change the number of vertices.
	
	\begin{proof}[Proof of \cref{thm: reduction to bounded instance}]
		First, we guarantee $n=2^j-1$ for some $j \in \mathbb{Z}$.
		If this does not hold initially, let $j$ be the smallest integer such that $n<2^j-1$. We add $2^j-n-1$ many vertices colocated at $s$.
		Note that this increases the number of vertices by at most factor $2$. 
		
		Next, we check whether the optimum total latency $\opt$ is zero by checking whether the strongly connected components of the subgraph containing all zero weight edges form a chain.
		If $\opt >0$, we change the distances in $c$ similar to \cite{quasi-poly}, so that each distance in the resulting asymmetric metric $c'$ is a positive integer bounded by a polynomial in $n$ and $\frac{1}{\epsilon}$ by losing a factor of $(1+\epsilon)$ in the approximation. 
		We start by computing a value $\gamma$ so that $\opt \le \gamma \le n^3 \opt$ by computing the smallest value so that the subgraph containing all edges of length at most $\gamma$ admits a walk starting in $s$ and visiting all vertices. 
		Clearly, $\opt \ge \gamma$, and by shortcutting this walk, $\opt \le n^3 \gamma$. Indeed, every edge in the shortcut path has length at most $n \gamma$.
		
		Next, we scale the length of each edge with a factor of $\frac{n^5}{\gamma \epsilon}$ and round it to the next positive integer.
		Note that the optimum total latency of the resulting instance can be bounded by $\opt' < \frac{n^5}{\gamma \epsilon} \opt + n^2$.
		Moreover, we cap the length of each edge at $\frac{2n^5}{\epsilon^2}$. 
		Note that any $\alpha$-approximate solution does not use an edge of length larger than
		\begin{equation*}
			\alpha \opt' < \alpha  \left(\frac{n^5}{\gamma \epsilon} \opt + n^2\right) \le \alpha \left(\frac{n^5}{\epsilon} + n^2\right) \le \frac{2n^5}{\epsilon^2} \enspace,
		\end{equation*}
		where we used $\epsilon \le 1$ and $\alpha \le \frac{1}{\eps}$ in the final inequality.
		Note that after both modifications, the resulting edge weights $c'$ still satisfy the triangle inequality.
		
		Finally, we insert a target $\target$ and add edges of length $1$ from each vertex in $V \cup \{s\}$ to $\target$ and edges of length $\frac{2n^5}{\epsilon}$ from $\target$ to each vertex in $V \cup \{s\}$.
		Note that this does not change the optimum solution value.
		Moreover, each solution to the resulting $\epsilon$-nice instance $(V \cup \{s,s'\},c')$ of \textsc{Directed Latency with Target} can be transformed to a solution of the instance $(V \cup \{s\},c')$ of \textsc{Directed Latency} of the same cost by deleting the edge to $\target$.
		Thus, by computing an $\alpha$-approximate solution to $(V \cup \{s,s'\},c')$, we obtain a solution $P=\{v_0,v_1,\dots,v_n\}$ to $(V \cup \{s\},c')$ of  total latency at most $\alpha \opt'$.
		
		Evaluated with respect to the original metric $c$, $P$ has total latency 
		\begin{align*}
			\sum_{i=1}^n (n+1-i) c(v_{i-1},v_i) &\le \sum_{i=1}^n (n+1-i) \frac{\gamma \epsilon}{n^5} c'(v_{i-1},v_i) \\
			&\le \frac{\gamma \epsilon}{n^5} \alpha \opt' \\ 
			&< \alpha \opt + \alpha \frac{\gamma \epsilon}{n^3} \\
			&\le \alpha (1+\eps) \opt \enspace, 
		\end{align*}
		where we used in the first inequality that $P$ does not use an edge for which we capped its length.
	\end{proof}
	
	\section{Separation Oracle for the \ref{eq:time-indexed-LP}}\label{appendixB}
	The separation oracle is very similar to the one presented in Appendix A in \cite{quasi-poly} but because of different notation, we adopt the proof here in our setting. Let $(x,z)$ be a proposed solution for the \ref{eq:time-indexed-LP}. Apart from the set-constraints for the polytope $Q(G_T)$, there are polynomially many constraints which can be checked directly. So it remains to show that we can separate the following constraints: 
	\begin{equation}\label{eq: cut-constraints}
		\displaystyle z(\delta^-(S))  \ge  \displaystyle \sum_{t' \in [t]} x_{v,t'}  \qquad \forall v \in V, t \in [T], S \subseteq V[G_T]: 
		(s,0) \notin S,
		\{v\} \times [t]  \subseteq S \enspace.
	\end{equation}
	
	Note that for each $v \in V$ and each $t \in [T]$, we can check whether \eqref{eq: cut-constraints} holds for each $S \subseteq V[G_T]$ with $(s,0) \notin S$ and $\{v\} \times [t]\subseteq S$ by computing a minimum $s$-$(\{v\} \times [t])$-cut in $G_T$ with arc capacities $z$.  
	
	Then, a separation oracle for \eqref{eq: cut-constraints} follows from computing minimum $s$-$(\{v\} \times [t])$-cuts in $G_T$ for each $v\in V$ and each $t\in [T]$ and check whether it is at least $\sum_{t' \in [t]} x_{v,t'}$ or not.
	
\end{document}